\documentclass[journal,12pt,draftcls,onecolumn]{IEEEtran}
\usepackage{graphicx}
\usepackage{caption}
\usepackage{subcaption}
\usepackage{amssymb}
\usepackage{mathrsfs}
\usepackage{cite}
\usepackage{amsmath}
\usepackage[table]{xcolor}
\usepackage{tabularx}
\usepackage{amsthm}
\usepackage{url}
\usepackage{soul}

\newcommand{\tran}{\mathrm{T}}
\newcommand{\bsm}{\begin{smallmatrix}}
\newcommand{\esm}{\end{smallmatrix}}
\newcommand{\bbm}{\begin{bmatrix}}
\newcommand{\ebm}{\end{bmatrix}}
\newcommand{\bpm}{\begin{pmatrix}}
\newcommand{\epm}{\end{pmatrix}}
\newcommand{\bom}{\begin{matrix}}
\newcommand{\eom}{\end{matrix}}
\newcommand{\bs}{\begin{small}}
\newcommand{\es}{\end{small}}
\newcommand{\rank}{\mathrm{rank}}
\newcommand{\diag}{\mathrm{diag}}

\newcommand{\ssp}[1]{\mathscr{#1}}      
\newcommand{\fld}[1]{\mathbb{#1}}       
\newcommand{\op}[1]{\mathcal{#1}}       
\newcommand{\sumbanach}[1]{\sum{#1}}
\newcommand{\spanset}[1]{\ \mathrm{span}\{#1\}\ } 

\newcommand{\pardiff}[2]{\frac{\partial {#1}}{\partial {#2}}}

\newcommand{\infd}{Inf-D }
\def\QEDclosed{\hfill\IEEEQEDclosed}
\renewcommand{\qed}{\QEDclosed}
\renewenvironment{proof}[1][\proofname]{\noindent\nobreakspace{\bfseries #1}:\;}{\qed\par}

\newcommand{\myhl}[1]{{#1}}

\newtheorem{theorem}{Theorem}
\newtheorem{lemma}{Lemma}
\newtheorem{corollary}{Corollary}
\newtheorem{remark}{Remark}
\newtheorem{definition}{Definition}

\begin{document}

\title{A Geometric Approach to Fault Detection and Isolation of Two-Dimensional (2D) Systems}

\author{Amir~Baniamerian$^{\mathrm{1}}$, Nader~Meskin$^{\mathrm{2}}$ and Khashayar Khorasani$^{\mathrm{1}}$
\thanks{*This publication was made possible by NPRP grant No. 4-195-2-065 from the
Qatar National Research Fund (a member of Qatar Foundation). The
statements made herein are solely the responsibility of the authors}
\thanks{$^{1}$A. Baniamerian and K. Khorasani are with the Department of Electrical and Computer Engineering,
        Concordia University, Quebec, Canada
        {\tt\small am\_bani@encs.concordia.ca} and {\tt\small kash@ece.concordia.ca}}%
\thanks{$^{2}$N. Meskin is with the Department of Electrical Engineering, Qatar University,
        Doha, Qatar
        {\tt\small nader.meskin@qu.edu.qa}}%
}

\maketitle

\begin{abstract}
In this work, we develop a novel fault detection and isolation (FDI) scheme for discrete-time two-dimensional (2D) systems that are represented by the Fornasini-Marchesini model II (FMII). This is accomplished by generalizing the  basic invariant subspaces including unobservable, conditioned invariant and unobservability subspaces of 1D systems to 2D models. These extensions have been achieved and facilitated by representing a 2D model as an infinite dimensional (Inf-D) system  on a Banach vector space, and by particularly constructing algorithms that compute these subspaces in a \emph{finite and known} number of steps. By utilizing the introduced subspaces the FDI problem is formulated and necessary and sufficient conditions for its solvability are provided. Sufficient conditions for solvability of the FDI problem for 2D systems using both deadbeat and LMI filters are also developed. Moreover, the capabilities and advantages of our proposed approach are demonstrated  by performing  an analytical comparison with the currently available 2D geometric methods in the literature. Finally, numerical simulations corresponding to an approximation of a hyperbolic partial differential equation (PDE)  system of a heat exchanger, that is mathematically represented as a 2D model, have also been provided. 
\end{abstract}

\begin{IEEEkeywords}
2D systems, Fornasini-Marchesini model, infinite dimensional systems, fault detection and isolation, geometric approach, invariant subspaces, LMI-based observer design, deadbeat observer.
\end{IEEEkeywords}

\IEEEpeerreviewmaketitle

\section{Introduction}
\IEEEPARstart{O}{ver} the past few decades, the problem of fault detection and isolation (FDI) of dynamical systems has increasingly received larger interest and attention from the control community \cite{IsermannBook2006} (c.f. to the references therein). The increasing complexity of human-made machines and devices, such as gas turbine engines and chemical processes, has necessitated the need for development of more complex and sophisticated FDI algorithms. Not with standing this requirement,  development of  FDI algorithms for systems that are governed by partial differential equations (PDE) has not been fully addressed and received extensive attention in the literature. One approach to investigate the FDI problem of PDEs relies on obtaining an approximate model of the system. First, the PDE system is approximated by a simple model (such as an ordinary differential equation (ODE) model), and then sufficient conditions for solvability of the FDI problem are derived based on this approximate model.

It is well-known that {\it parabolic PDE} systems can be approximated by  ODE representations. These systems can be approximated through application of finite element methods  where sufficient conditions can then be derived by using singular perturbation theory \cite{Christo_Book}. The FDI problem of parabolic PDEs has been addressed by using the corresponding approximate models in \cite{ACC2012, Davis_Journal}.
On the other hand, by discretizing through  spatial coordinates, one can approximate  {\it hyperbolic PDE} systems by  ODE models. However, unlike parabolic PDE systems, one cannot apply model decomposition, order reduction and singular perturbation theory to hyperbolic PDE systems \cite{Christofer_Hyperbolic}. Moreover, the order of the resulting approximate ODE systems can be dramatically high. Therefore, researchers have investigated hyperbolic PDE systems by using  other formal methods such as the theory of semigroups \cite{Curtain_Book} and  backstepping methods \cite{Krstic_Hyperbolic}.

As shown in \cite{HyperPDE_2D}, a {\it single} hyperbolic PDE system can be approximated by using two-dimensional (2D) systems. In \cite{ACC2013}, we have shown that this approximation is also applicable to a {\it system} of hyperbolic PDEs. Moreover, as shown in \cite{Parabolic3D}, parabolic PDE systems can be approximated by three-dimensional (3D) Fornasini and Marchesini model II (FMII)  representations \cite{Kaczorek_Book, 2DRef, Stability2D2013, FMinBook}. We have shown in the Appendix that it is straightforward to extend the results of our work to 3D systems as well (these results are not included here due to space limitations). Consequently, our proposed methodology in this work can also be applied to parabolic PDE systems. There are only a few results on FDI of 2D systems in the literature, such as those by using dead-beat observers \cite{BisiaccoMultiDim} and parity equations \cite{FMResidual, BisiaccoLetter}. On the other hand, the FDI of hyperbolic PDE systems (by using \underline{even an approximate model}) has \underline{not} yet been adequately addressed in the literature.

It should be noted that 2D system theory has other applications in the control field. For example, a class of discrete-time linear repetitive processes can be modeled by  2D systems. These  processes play important roles in tracking control and robotics, where the controlled system is required to perform a periodic task with  high precision (refer to \cite{RepitativeProcess_Book} for more details on repetitive systems). One of the main approaches to control linear repetitive processes is the iterative learning control (ILC)\cite{RepitativeProcess_Book}. Since the ILC problem can be formulated as a control problem in the 2D system theory \cite{ILC_2D_Tac}, 2D systems have been increasingly applied to spatio-temporal and repetitive process control problems in the literature.

The FDI of 1D linear, time-invariant systems has been extensively investigated during the past few decades \cite{IsermannBook2006} (and the references therein). The geometric approach \cite{Mass_Thesis, Massoumnia1986,  Massoumnia1989} has provided a valuable tool for investigating the FDI problem of a large class of dynamical systems such as linear time-invariant \cite{Massoumnia1989}, Markovian jump systems \cite{DrMeskin_TacFullPaper}, time-delay systems \cite{DrMeskin_Delay}, linear impulsive systems \cite{DrMeskin_Impulsive}, and parabolic PDE  systems \cite{ACC2012}. Moreover, the geometric approach is also extended to affine nonlinear systems in \cite{IsidoriFDI}.  Furthermore, hybrid geometric FDI approaches for linear and nonlinear systems have also been provided in \cite{MKR_SysTech2010} and \cite{MKR_Nonlin2010}, where a set of residual generators are equipped with a discrete event-based system fault diagnoser to solve the FDI problem.

Motivated by the above discussion, in this paper we investigate the FDI problem of 2D systems and apply the results to a 2D approximate model of a hyperbolic PDE system. As stated earlier, a hyperbolic PDE system can be approximated by  1D systems, where the order of the approximate 1D system will significantly increase by decreasing the gridding size \cite{Christofer_Hyperbolic}. We also provide an example where faults in the 1D approximate model are not isolable, whereas one can detect and isolate  faults by applying the 2D approximate model and by using our proposed methodology.

Recently, the geometric theory of 2D  systems has attracted much interest, where basic concepts such as conditioned invariant and controllable subspaces are studied in detail for the Fornasini and Marchesini model I (FMI) \cite{ ntogramatzidis2012GeometryArticle, ntogramatzidis2012Siam}. The hybrid 2D systems have also been investigated from the geometric point of view in \cite{Valcher2013}. The geometric FDI approach for 2D systems, for the first time, was addressed in \cite{ACC2013}, where invariant subspaces of the Roesser model are defined and the FDI problem is formulated based on these subspaces.

Two-dimensional (2D) systems have  been extensively investigated from a system theory point of view \cite{Kaczorek_Book, 2DRef, Stability2D2013, FMinBook}. Particularly, system theory concepts such as stability \cite{2DLyapunov, Stability2D2013}, controllability \cite{Controlability2D}, observability\cite{Valcher2013}, and state reconstruction \cite{2DKalman} have been investigated in the literature. However, due to complexity of 2D systems, unlike one-dimensional (1D) systems,  there are various definitions that are introduced for controllability and observability properties. Not surprisingly, the duality between the observability and controllability does not hold in 2D systems.
In this paper, we investigate observability of 2D systems from a \underline{{\it new}} geometric point of view  which has its roots in system theory of infinite dimensional (Inf-D) systems.

Compared to the results reported in \cite{ACC2013} and \cite{ACC2014}, we have specific generalization and novel contributions in this work. We \underline{first} investigate the Fornasini and Marchesini model II (FMII) as an Inf-D system that allows us to deal with Inf-D subspaces (albeit with a finite dimensional (Fin-D) representation). The invariant subspaces and the corresponding algorithms  that are introduced in \cite{ACC2013} for the Roesser model are then \underline{generalized} to the 2D FMII systems. It is worth noting that although the introduced subspaces in this work are Inf-D, the corresponding algorithms for constructing the subspaces converge in a finite and known number of steps.

In addition, in \cite{ACC2013} only sufficient conditions for solvability of the FDI problem of the Roesser model were provided. As shown in \cite{ACC2014}, the invariance property  of an unobservable subspace is a generic property of 2D systems. We first derive a \underline{single} necessary and sufficient condition for detectability and isolability of faults (that is formally introduced in the next section). It is shown that this condition also necessary for solvability of the FDI. Furthermore, in \cite{ACC2013} by utilizing the existence of an LMI-based observer  only sufficient conditions for solvability of the Roesser model FDI problem were provided. In other words, the procedure to design the observer gains is not provided in \cite{ACC2013,ACC2014}. However, in our paper, we derive both necessary and sufficient conditions, where sufficient conditions are based on  (a) an ordinary, (b) a delayed deadbeat, and (c) an LMI-based 2D Luenberger filters. Moreover,  we develop a procedure to design LMI-based filter gains. 

It must be noted that recently related work has appeared in \cite{Malek_3DFDI} and \cite{Malek_3DFDIConf}. These two papers investigate the FDI problem of three-dimensional (3D) FMII models. Although, a geometric FDI methodology is also developed in \cite{Malek_3DFDI}, our work is distinct and unique from \cite{Malek_3DFDI} in \underline{{\it three}} main perspectives:
\begin{enumerate}
	\item The approach proposed  in \cite{Malek_3DFDI} is based on the results of \cite{Massoumnia1986}, whereas our approach is based on the generalization of the results of \cite{Massoumnia1986} as reported in \cite{Massoumnia1989} (the results in \cite{Massoumnia1989} are more general than \cite{Massoumnia1986}) for 2D models (for more information refer to Remark \ref{Rem:GeneralFilter}).
	\item In \cite{Malek_3DFDI}, necessary and sufficient conditions for solvability of the FDI problem were derived for a \underline{subclass of detection filters} where it was assumed that the output map of the detection filters and that of the system are identical. However,  in our work here we consider a \underline{general class of detection filters} for the residual generation and relax this condition.
	\item As shown in Section \ref{Sec:FDIComparison}, the observability property of the 2D model is a fundamental requirement and assumption in \cite{Malek_3DFDI} (although it is stated in \cite{Malek_3DFDI} that this assumption was made  for simplicity of their presentation). However, our proposed solution \underline{does not} require this condition and assumption, and consequently our approach leads to a less restrictive solution.
\end{enumerate}

Another approach that was developed in the literature \cite{BisiaccoMultiDim,BisiaccoLetter} has its roots in 2D deadbeat observers \cite{Bisiacco_Obs}. In \cite{BisiaccoMultiDim}, the FDI  problem is investigated by using polynomial matrices and unknown input deadbeat observers, where the right zero primeness of the 2D Popov-Belevitch-Hautus (PBH) matrix (this is reviewed in Subsection \ref{Sec:DeabeatObs})  is a necessary condition for solvability of the FDI problem. In \cite{BisiaccoLetter}, this condition was relaxed and  necessary and sufficient conditions that are based on an extended parity equation approach were obtained.

To provide a fair and comprehensive comparison with the currently available result in \cite{Malek_3DFDI}, in this work it is shown that on one hand the solvability of the FDI problem by using the method in \cite{Malek_3DFDI} is also sufficient to accomplish the FDI task by using our proposed approach. On the other hand, there are certain 2D systems that are \underline{not solvable} by the approach in \cite{Malek_3DFDI}, however our proposed approach \underline{can both detect and isolate} the faults. Also, by comparing our proposed results with the algebraic-based methods in \cite{BisiaccoMultiDim,BisiaccoLetter}, where they can solve and derive necessary and sufficient conditions for solvability of the FDI problem, we will highlight and emphasize  two important considerations as follows:
\begin{enumerate}
	\item The algebraic methods, in contrast to our geometric approach,  need a closed-form and analytical solution to certain polynomial matrix equations (in two variables). However, our proposed approach is derived and solved by using commonly available and relatively straightforward numerical methods.
	\item As shown subsequently in Section \ref{Sec:FDIComparison}, there are certain examples, where the necessary conditions in \cite{BisiaccoMultiDim,BisiaccoLetter} are not satisfied. However, our proposed approach can \underline{both detect and isolate} these fault scenarios.
\end{enumerate}
To summarize, the \underline{main contributions} of this paper can be highlighted as follows:
\begin{enumerate}
	\item By reformulating  2D models as  Inf-D systems, the invariance property of the unobservable subspace is investigated (this is provided in Section \ref{Sec:InvSpace}), where an Inf-D unobservable subspace is also introduced. This result enables one to formally address the  solvability of the FDI problem without restriction on the initial conditions (unlike in \cite{ACC2013, ACC2014} where a restrictive assumption that the unobservable subspace is Fin-D  is imposed). Since 2D systems are Inf-D dynamical system, our proposed Inf-D representation and framework enables one to address the FDI problem in its most general scenario than that in \cite{ACC2013} and \cite{ACC2014}.
	\item Two important Inf-D invariant subspaces (namely, the conditioned invariant and the unobservability) are introduced for the FMII models. Although, these subspaces are Inf-D, we provide explicit algorithms that can be invoked to compute these subspaces in a finite and known number of steps.
	\item The FDI problem of 2D systems is formulated in terms  of the above introduced invariant subspaces, and necessary and sufficient conditions for its solvability are derived and formally analyzed.
	\item A novel procedure is developed for designing an observer (also known as a detection filter)  by utilizing the linear matrix inequalities (LMI) technique.
	\item Three sets of sufficient conditions for solvability of the FDI problem by utilizing the ordinary, the delayed deadbeat, as well as our proposed LMI-based observers (detection filters) are also provided.
	\item Analytical comparisons between our proposed approach and the one in \cite{Malek_3DFDI} are presented. We show that the sufficient conditions in\cite{Malek_3DFDI} are also sufficient for solvability of the FDI problem by using our  approach. However, an example is provided that shows our method can both detect and isolate faults, whereas the approach in  \cite{Malek_3DFDI}  \underline{cannot} be used. In other words, it is shown that if the method in the above literature can accomplish the FDI task, our proposed approach can also accomplish this task. However, if our scheme \underline{cannot} achieve the fault detection and isolation goals for a given system, then it is guaranteed that \underline{the schemes} in \cite{Malek_3DFDI} cannot also achieve these goals. Moreover, there are 2D systems  where our approach can achieve the FDI objectives whereas the results in the  literature \underline{cannot} solve the FDI problem.
	\item Our proposed methodology and strategy is applied to an important application area of a heat exchanger (a hyperbolic PDE system), where it is shown that one can \underline{simultaneously} detect and isolate two different faults namely, the leakage and the fouling faults.
\end{enumerate}

The remainder of the paper is organized as follows. The {\it preliminary} results including the Inf-D representation, the FDI problem formulation, the 2D deadbeat observer and the 2D Luenberger observers (detection filters) are presented in Section \ref{Sec:Backgground}. The unobservable subspaces of the FMII 2D model are introduced in Section \ref{Sec:InvSpace}. The geometric property of these subspaces and the invariant concept of the FMII model are also presented in Section \ref{Sec:InvSpace}. In Section \ref{Sec:FDI}, necessary and sufficient conditions for solvability of the FDI problem are derived and developed. Analytical comparisons between our proposed approach and the available \emph{geometric methods} in the literature, namely  \cite{Malek_3DFDI} and \cite{Malek_3DFDIConf} are provided in this section. Furthermore,  numerical comparisons with both \textit{geometric and algebraic methods} in \cite{Malek_3DFDI,Malek_3DFDIConf,BisiaccoMultiDim,BisiaccoLetter} are presented in this section. Simulation results for the FDI problem of a heat exchanger that is expressed as a PDE system are conducted in Section \ref{Sec:Simulation}. Finally, Section \ref{Sec:Conc} concludes the paper and provides suggestions for future work.

{\bf Notation:}
In this work, $\mathscr{A},\mathscr{B},...$ are used to denote subspaces. For a given vector $L$, the subspace $\spanset{L}$ is denoted by $\ssp{L}$. The inverse image of the subspace $\ssp{V}$ with respect to the operator $A$ is denoted by $A^{-1}\ssp{V}$. The block diagonal matrix $\bbm A&0\\0 &B\ebm$ is denoted by $\diag(A,B)$. The real, complex, integer and positive integer numbers are denoted by $\fld{R}$, $\fld{C}$, $\fld{Z}$ and $\fld{N}$, respectively. $\underline{\fld{N}}$ denotes the set $\fld{N}\cup \{0\}$. In this paper, we deal with infinite dimensional (Inf-D) subspaces and vectors. An Inf-D vector is designated by the bold letters ${\bf x},{\bf y},\cdots$. The Inf-D subspace $\cdots\oplus\ssp{V}\oplus\ssp{V}\oplus\cdots$ is denoted by $\oplus\ssp{V}$, where $\ssp{V}\subseteq\fld{R}^n$. Let ${\bf x} =(\cdots,x_{-1}^\tran,x_0^\tran,x_1^\tran,\cdots)^\tran\subseteq\oplus\ssp{V}$ and  $|{\bf x}|_\infty = \underset{i\in\fld{Z}}{\text{sup}} |x_i|$, where $x_i\in\ssp{V}$. The vector space $\ssp{V}_\infty = \sumbanach{\ssp{V}}$ is defined as $\{{\bf x}|{\bf x}\in\oplus\ssp{V}\;\mathrm{and}\; |{\bf x}|_\infty<\infty\}$. It can be shown that $\ssp{V}_\infty$ is a Banach (but not necessarily Hilbert) space.  Let $i,j\in\fld{Z}\cup\{-\infty,\infty\}$ and $j\geq i$. The Inf-D vector ${\bf x}_i^j\in\oplus\ssp{V}$ is expressed as ${\bf x}_i^j=[\cdots,0, 0, x_i^\tran,\cdots, x_j^\tran,0 ,0,\cdots]^\tran$, where $x_\ell\in\ssp{V}$ for all $i\leq\ell\leq j$, and associated with $\textbf{x}_{-\infty}^\infty$ we simply  use $\textbf{x}$. The other notations are provided within the text of the paper as appropriate.
\section{Preliminary Results}\label{Sec:Backgground}
In this section, we first review 2D systems and their various representational models. Subsequently,  a 2D system is expressed as an infinite dimensional (Inf-D) system that allows one to geometrically analyze the unobservable subspaces (this is to be defined and specified in the next section). The FDI problem is also formulated in this section. Moreover, we review the 2D Popov-Belevitch-Hautus (PBH) matrix and 2D deadbeat observers in this section. Finally, an LMI-based approach is introduced to design a 2D Luenberger observer (also known as a detection filter) for 2D systems.
\subsection{Discrete-Time 2D Systems} \label{Sec:Dis_2D}    
2D models can be used for representing a large class of problems such as approximating hyperbolic PDE systems \cite{ACC2013, HyperPDE_2D}, image processing  and digital filtering \cite{Roess}. System theory concepts such as observability, controllability and feedback stabilization have also been investigated in the literature  for 2D systems \cite{Kaczorek_Book, FMinBook, ACC2013, ntogramatzidis2012Siam, Valcher2013}.
There are various models that are adopted in the literature for 2D systems including the Rosser model \cite{Roess}, the Fornasini-Marichesini model I (FMI) and model II (FMII) \cite{FMinBook,Kaczorek_Book}. The FMI can be formulated as a Roesser model and the Roesser model is a special case of the FMII model \cite{Kaczorek_Book}. In this work, we consider and concentrate on the FMII model, and consequently our results are also derived for this general class of 2D systems. 

Consider the following FMII model \cite{FMinBook}, 
\begin{equation}\label{Eq:FMII}
	\begin{split}
		x(&i+1,j+1) = A_1x(i,j+1) + A_2x(i+1,j) + B_1u(i,j+1) + B_2u(i+1,j)\\&+
		\sum_{k=1}^p L_{k}^1 f_k(i,j+1) + \sum_{k=1}^p L_{k}^2 f_k(i+1,j),\\
		y(&i,j) = C x(i,j), \; i,j\in\fld{Z},
	\end{split}
\end{equation}
where $x\in\fld{R}^n$, $u\in\fld{R}^m$, and $y\in \fld{R}^q$ denote the state, input and output vectors, respectively. The fault signals and the corresponding fault signatures are designated by$f_k$, $L_k^1$ and $L_{k}^2$, respectively. Also, $p$ denotes the number of faults in the system. Since in this work all the introduced invariant subspaces are based on the operators $A_1$, $A_2$ and $C$, we designate the system \eqref{Eq:FMII} by the triple ($C$,$A_1$,$A_2$). 
\begin{remark}\label{Rem:onFMII-General}
	Note that system \eqref{Eq:FMII} represents and captures the presence of \underline{both} actuator and component faults. To represent sensor faults, one can augment the sensor dynamics and model the sensor faults as actuator faults in the augmented system (for a complete discussion on this issue refer to \cite{Mass_Thesis} - Chapters 3 and 4). Also, it should be pointed out that the fault signal $f_k$ affects the system through two different fault signatures $L_{k}^1$ and $L_{k}^2$. An alternative fault model could have been expressed according to the following representation,
	\begin{align}\label{Eq:FMII_SingFault}
	x(i+1,j+1) = &A_1x(i,j+1) + A_2x(i+1,j) + B_1u(i,j+1) + B_2u(i+1,j)+
	\sum_{k=1}^p L_{k} g_{k}(i,j),\notag\\
	y(i,j) = &C x(i,j).
	\end{align}
	Model \eqref{Eq:FMII} is more general than the one given by equation  \eqref{Eq:FMII_SingFault}. This is due to the fact that by denoting  $f_k(i+1,j) = g_k(i,j)$ for all $k = 1,\cdots,p$, one can represent the model \eqref{Eq:FMII_SingFault} as in the model \eqref{Eq:FMII}. \qed
\end{remark}

Let us now consider the Roesser model \cite{Roess} which is expressed as
\begin{equation}\label{Eq:RoesserModel}
\begin{split}
\bbm r(i+1,j)\\s(i,j+1)\ebm = &\bbm A_{11} &A_{12}\\ A_{21} &A_{22}\ebm\bbm r(i,j)\\s(i,j)\ebm +\bbm B_{11}\\ B_{21}\ebm u(i,j)+\sum_{k=1}^p L_k f_k(i,j),\\
y(i,j) = &C\bbm r(i,j)\\s(i,j)\ebm,
\end{split}
\end{equation}
and where $\bbm r^\tran &s^\tran\ebm^\tran\in \fld{R}^n$ represents the state, and the variables $u$, $y$, $f_k$ and $L_k$ are defined as in equation \eqref{Eq:FMII}. By defining
\begin{equation}
\begin{split}
x &= \bbm r\\s\ebm, A_1 = \bbm A_{11} &A_{12}\\ 0 &0\ebm, A_2 = \bbm 0 &0\\ A_{21} &A_{22}\ebm,
B_1 = \bbm B_{11}\\0\ebm, B_2 = \bbm 0\\ B_{21}\ebm,
\end{split}
\end{equation}
one can formulate the Roesser model \eqref{Eq:RoesserModel} as in equation \eqref{Eq:FMII}. In this paper, we assume that $A_1$ and $A_2$ in model \eqref{Eq:FMII} are not necessarily commutative (i.e. $A_1A_2 \neq A_2A_1$), and hence, the results that are subsequently developed  can also be applied to the Roesser model \eqref{Eq:RoesserModel}. It should be noted that the commutativity of  $A_1$ and $A_2$ is a strong condition that renders the results in\cite{Valcher2013} (where it is assumed that $A_1$ and $A_2$ are commutative) not applicable to the system \eqref{Eq:RoesserModel}.

In this work, we will investigate and develop FDI strategies for the model \eqref{Eq:FMII}. It is assumed that $A_1$ and $A_2$ in model \eqref{Eq:FMII} are not necessarily commutative (i.e. $A_1A_2 \neq A_2A_1$), and hence, the results that are subsequently developed  can also be applied to the Roesser model. It should be emphasized that the commutativity of  $A_1$ and $A_2$ is a strong condition that renders the results in\cite{Valcher2013} (where  $A_1$ and $A_2$ are assumed to commutate) not applicable to Roesser systems.

\subsection{Infinite Dimensional (\infd) Representation} \label{Sec:IDRep}
In this subsection, we reformulate the 2D model \eqref{Eq:FMII} as an Inf-D system that will be used to derive the invariance property of  unobservable subspaces (for details refer to Section \ref{Sec:UnObservable}).

Consider the fault free system \eqref{Eq:FMII}, that is with $f_k\equiv0$, and with zero input (we are mainly interested in the unobservable subspaces and do not need to be concerned with the control inputs in the FDI problem). By considering  ${\bf x}(k)=(\cdots,x(-1+k,1)^\tran,x(k,0)^\tran,x(1+k,-1)^\tran,\cdots)^\tran\in\sum\fld{R}^n$, it can be shown that under the above conditions the system \eqref{Eq:FMII} can be represented as,
\begin{equation}\label{Eq:IDRep}
	\begin{split}
		{\bf x}(k+1) &= \op{A}{\bf x}(k),\;\;\;\; k\in\underline{\fld{N}} \\
		{\bf y}(k) &=\op{C}{\bf x}(k),
	\end{split}
\end{equation}
where ${\bf x}(k)\in\op{X}=\sumbanach{\fld{R}^n}$, ${\bf y}(k)=(\cdots,y(-1+k,1)^\tran,y(k,0)^\tran,y(1+k,-1)^\tran,\cdots)^\tran\in\sumbanach{\fld{R}^q}$, and $\op{A}$ is an Inf-D matrix with $A_1$ and $A_2$ as diagonal and upper diagonal blocks, respectively, with the remaining elements set to zero, and $\op{C} = \diag(\cdots,C,C,\cdots)$. In other words, we have,
\begin{equation}
	\op{A} = \bbm \ &\ddots &\ddots &\cdots & &\cdots\\
	\cdots &0 &A_1 &A_2 &0 &\cdots\\
	\cdots &0 &0 &A_1 &A_2 &\cdots\\
	\cdots & &\cdots & &\ddots &\ddots \ebm ,
	\op{C} = \bbm \ddots & &\cdots & &\cdots\\
	\cdots &0 &C &0 &\cdots\\
	\cdots &0 &0 &C &\cdots \\
	\cdots & &\cdots & &\ddots\ebm												
\end{equation}

Note that since we invoke an Inf-D representation to investigate an unobservable subspace, and where this subspace is defined by only $\op{A}$ and $\op{C}$, therefore for sake of  presentation simplicity,  an Inf-D system is used that has no fault and zero input.

There are various formulations for the initial conditions  of the FMII model \eqref{Eq:FMII}. These are based on the separation set that is introduced in \cite{FM_SeparationSet}. There are two separation sets that are commonly used in the literature.  In the first formulation
the initial conditions are denoted by ${\bf x}(0)=(\cdots,x(-1,1)^\tran,x(0,0)^\tran,x(1,-1)^\tran,\cdots)^\tran\in\sum\fld{R}^n$ \cite{FMinBook} (this is compatible with the model \eqref{Eq:IDRep}). The second formulation is expressed as $x(i,0)=h_1(i)$ and $x(0,j)=h_2(j)$, where $h_1(i),h_2(j)\in \fld{R}^n$ and $i,j\in\fld{N}$ \cite{Kaczorek_Book}.
The second formulation is more compatible with applications (particularly, in case that the system \eqref{Eq:FMII} is an approximate model of a PDE system - refer to Section \ref{Sec:Simulation}). It will be shown subsequently that  since we derive the conditions based on invariant unobservable subspace (this is formally defined in the next section), our proposed methodology is applicable to \underline{both} initial condition formulations. In other words, we use the Inf-D  representation  to only show the results and evaluate the developed algorithms. However, to apply our results there is no need to deal with Inf-D systems and subspaces, and therefore, one can apply our proposed methods to 2D systems corresponding to both initial condition formulations. 

\myhl{We start with the first formulation of the initial} conditions subject to the boundedness assumption (this is, ${\bf x}(0)\in\sum\fld{R}^n$). However, as shown in Section \ref{Sec:A12Inv}, our proposed results also hold for  the second initial condition formulation.

As stated in the Notation section, it can be  shown that $\op{X}$ defined for equation \eqref{Eq:IDRep} is an Inf-D Banach space. The system theory corresponding to Inf-D systems is more significantly challenging than Fin-D system theory (1D systems) (refer to \cite{Zwart_Book}). However, as shown subsequently, the operator $\op{A}$ is bounded and consequently, one can readily extend the result of 1D systems to the system \eqref{Eq:IDRep} \cite{Curtain_Book,Zwart_Book}. Let us first define the notion of bounded operators.
\begin{definition}\label{{Def:BoundedOp}}\cite{Curtain_Book}
	Consider the operator $\op{A}:\op{X}_1\rightarrow\op{X}_2$, where $\op{X}_1$ and $\op{X}_2$ are Banach vector spaces  with the norms $|\cdot|_1$ and $|\cdot|_2$, respectively. The operator $\op{A}$ is bounded if there exists a real number $G$ such that $|\op{A}{\bf x}|_2\leq G|{\bf x}|_1$ for all ${\bf x}\in\op{X}_1$.
\end{definition}
\begin{lemma}\label{Lem:Boundedness}
	The operator $\op{A}$ as defined in the Inf-D system \eqref{Eq:IDRep} is bounded.
\end{lemma}
\begin{proof}
	Let $G = 2\max(|A_1|,|A_2|)$, where $|A_i|$ denotes the norm of $A_i$ and ${\bf x}=(x_k)_{k\in\fld{Z}}\in\op{X}$. It follows readily that $|\op{A}{\bf x}|_\infty=\underset{k\in\fld{Z}}{\text{sup}} |A_1x_k+A_2x_{k+1}|\leq\underset{k\in\fld{Z}}{\text{sup}} G \max(|x_k|,|x_{k+1}|)=G \underset{k\in\fld{Z}}{\text{sup}} |x_k|$.
	Therefore, $|\op{A}{\bf x}|_\infty\leq G|{\bf x}|_\infty$. This completes the proof of the lemma.
\end{proof}
The above lemma enables one to now formulate the unobservable subspace of the 2D system \eqref{Eq:FMII} in a geometric framework (for details refer to Section \ref{Sec:InvSpace}) based on the operator $\op{A}$ (and consequently, in terms of $A_1$ and $A_2$).

\begin{remark}
	 Although, in \cite{conte1988GeometryArticle} all the results such as the controlled invariant subspaces are presented on $\fld{R}^n$, the developed approach in \cite{conte1988GeometryArticle} has its roots in the theory of systems over rings. In this paper, we propose an alternative approach that is based on Inf-D systems that are  defined on a Banach vector space. Similar to \cite{conte1988GeometryArticle}, our proposed methodology including the algorithms and the conditions for solvability of the FDI problem can also be addressed in a Fin-D scheme. However, as shown in the literature the duality property does not hold for 2D systems \cite{FMMinimalRealization,FMinBook}. Therefore, by simply invoking duality the results of this paper cannot be derived from those in \cite{conte1988GeometryArticle}.
\end{remark}

\subsection{The FDI  Problem of 2D FMII Model}\label{Sec:FDIProb}
In this subsection,  we formulate the FDI problem for the 2D system \eqref{Eq:FMII}. In this paper, without loss of any generality, it is assumed that the system \eqref{Eq:FMII} is subject to two faults, and therefore we construct two residuals such that each one is sensitive to only one fault and is decoupled from the other.

More precisely, consider the faulty FMII model \eqref{Eq:FMII}. The solution to the FDI problem of the 2D FMII system can be stated as that of generating two residuals $r_k(i,j), k\in\{1,2\}$ such that,
\begin{subequations}\label{Eq:FDIP}
	\begin{align}
		&\forall u , f_2 \  \mathrm{and} \ f_1=0 \;\;\; \mathrm{then} \  r_1\rightarrow 0 \nonumber, \\ &\mathrm{and}\;\mathrm{if} \ \ f_1\neq 0 \;\;\; \mathrm{then} \ r_1\neq  0  \label{FDI_InObs1},\\
		&\forall u,  f_1  \  \mathrm{and} \ f_2=0 \;\;\; \mathrm{then} \  r_2\rightarrow 0 \nonumber, \\ &\mathrm{and}\;\mathrm{if} \ \ f_2\neq 0 \;\;\; \mathrm{then} \ r_2\neq  0 \label{FDI_InObs2}.
	\end{align}
\end{subequations}

The above residuals are to be constructed by employing fault detection filters. For the 2D system
\eqref{Eq:FMII}, we consider the following FMII-based {\it fault detection filter},
\bs
\begin{align}\label{Eq:Filter}
	\hat{\omega}(i+1,j+1) &= F_1\hat{\omega}(i,j+1) + F_2\hat{\omega}(i+1,j)+ K_1u(i,j+1) + K_2u(i+1,j) + E_1y(i,j+1) + E_2y(i+1,j),\notag\\
	r_k(i,j) &= M\hat{\omega}(i,j)-Hy(i,j),
\end{align}
\es
where $\hat{\omega}(i,j)$ denotes the state of the filter and is used to define the residual signal $r_k(i,j)$. The solution to the FDI problem is now reduced to that of selecting the filter gains $F_1$, $F_2$, $K_1$, $K_2$, $E_1$, $E_2$, $M$ and $H$ corresponding to the filter \eqref{Eq:Filter}.
\begin{remark}\label{Rem:GeneralFilter}
	The detection filter \eqref{Eq:Filter} can be selected as a full-order ($H=I$) or as a partial-order ($ker H\neq 0$) 2D Luenberger observer.
	As shown subsequently in Section \ref{Sec:FDI}, this level of generality allows one to analytically compare our proposed methods with those results reported in \cite{Malek_3DFDI}.\qed
\end{remark}

\begin{remark}\label{Rem:ProbName}
	In this paper, we investigate the FDI problem by employing two main steps, namely (i) decoupling the faults, and (ii) designing  filter gains for each fault. The first step for decoupling $f_1$  addresses  the existence of three maps $D_1$, $D_2$ and $H$, such that the fault $f_2$ signatures $L_2^1$ and $L_2^2$  are members of the unobservable subspace (defined in the next section) of the system ($H_kC$, $A_1+D_1C$, $A_2+D_2C$). The same terminology is used to decouple $f_2$. Moreover, the second step is mainly concerned with existence of the filter \eqref{Eq:Filter} such that stability of the error dynamics is guaranteed. In this paper, if the first step is solvable for the fault $f_i$ we say that $f_i$ is detectable and isolable.  Finally, it is stated that there is a solution to the FDI problem if for all the fault signals $f_i$  \underline{both} steps above are solvable.
\end{remark}

\subsection{Deadbeat Observers}\label{Sec:DeabeatObs}
In Section \ref{Sec:FDI},  necessary and sufficient conditions for solvability of the FDI problem are derived. We provide  sufficient conditions for accomplishing the FDI task by using a delayed deadbeat detection filter and an ordinary (i.e., without a delay) deadbeat observer (refer to the subsequent Corollaries \ref{Col:FDI_DelayedDeadBeat_Suff} and \ref{Col:FDI_DeadBeat_Suff}). Towards these end, in this section we formally define a (delayed) deadbeat filter. For a comprehensive discussion on 2D deadbeat observers refer to \cite{Bisiacco_Obs,BisiaccoLetter}.

Consider the system \eqref{Eq:FMII} under the fault free situation. A (delayed) deadbeat observer is constructed according to,
\bs
\begin{align}\label{Eq:DeadbeatObs}
	z(i+1,j+1) =& F_1z(i,j+1) + F_2z(i+1,j)+ K_1u(i,j+1) + K_2u(i+1,j) + E_1y(i,j+1) + E_2y(i+1,j),\notag\\
	\hat{x}(i,j) =& L_1z(i,j)+L_2y(i,j),
\end{align}
\es
where $u$ and $y$ are the input and output signals as defined in the system \eqref{Eq:FMII}. Note that since  the output is assumed to not be directly affected by the input signal, $\hat{x}(i,j)$ is only a linear combination of $z(i,j)$ and $y(i,j)$.  If there exists a number $N_0$ such that $\hat{x}(i,j)=x(i,j)$ for all $i+j>N_0$, the filter \eqref{Eq:DeadbeatObs} is designated as an \underline{ordinary (without a delay) deadbeat observer}. On the other hand, if there exist \underline{non-negative} integers $n_1$ and $n_2$ such that $n_1+n_2>0 $ and $\hat{x}(i,j)=x(i-n_1,j-n_2)$, the observer \eqref{Eq:DeadbeatObs} is  designated as a \underline{{\it delayed} deadbeat observer}. The necessary and sufficient conditions for existence of a (delayed) deadbeat observer are specified in the following theorem \cite{Bisiacco_Obs, BisiaccoLetter}.
\begin{theorem}\label{Thm:DeadbeatObsExistance}
	Consider the 2D system \eqref{Eq:FMII} under a fault free situation and the following 2D Popov-Belevitch-Hautus (PBH)  matrix,
	\begin{equation}\label{Eq:PBHMatrix}
		PBH(z_1,z_2) = \bbm I-z_1A_1-z_2A_2\\ C\ebm
	\end{equation}
	where $A_1$, $A_2$ and $C$ are defined as in equation \eqref{Eq:FMII}, and $z_1,z_2\in\fld{C}$. Then,
	\begin{enumerate}
		\renewcommand{\labelenumi}{(\roman{enumi})}
		\item there exists a delayed  deadbeat observer if and only if $\rank(PBH(z_1,z_2))=n$ for all $z_1,z_2\in\fld{C}-\{0\}$ (that is, $PBH(z_1,z_2)$ is right monomic) \cite{BisiaccoLetter} (Section 3) and
		\item there exists an ordinary deadbeat observer if and only if $\rank(PBH(z_1,z_2))=n$ for all $z_1,z_2\in\fld{C}$ (that is, $PBH(z_1,z_2)$ is right zero prime) \cite{Bisiacco_Obs}.\qed
	\end{enumerate}
\end{theorem}
We use the above theorem in Section \ref{Sec:FDI} to show the existence of deadbeat filters to derive sufficient conditions for solvability of the FDI problem (the subsequent Corollaries \ref{Col:FDI_DelayedDeadBeat_Suff} and \ref{Col:FDI_DeadBeat_Suff}).
\subsection{LMI-based Observer (Detection Filter) Design}\label{Sec:Observer}
As shown in \cite{Bisiacco_Obs}, design of a deadbeat observer \eqref{Eq:DeadbeatObs} requires that one works with polynomial matrices (this is not always a straightforward process). In this subsection, we address the design process for the FMII system observer, or the detection filter gains, by using linear matrix inequalities (LMI). These results will be used to explicitly design a 2D Luenberger detection filter (that can also be formulated as in equation \eqref{Eq:Filter}) subsequently in Section \ref{Sec:FDI} for the  purpose of accomplishing the solution to the FDI problem. 

In order to show the asymptotic stability of the state estimation error dynamics, one needs to apply the following stability lemmas.

\begin{lemma}\label{Lem:2DLyapunov}\cite{2DLyapunov}
	The 2D FMII system \eqref{Eq:FMII} (under the fault free situation) is asymptotically stable if there exist two symmetric positive definite matrices $R_1,R_2\in\fld{R}^{n\times n}$ such that,
	\begin{equation}\label{Eq:2DLyapunov}
		A_c  \triangleq A^\tran (R_1+R_2) A- R <0,
	\end{equation}
	where $A = [A_1~~A_2]$ and $R=\diag(R_1,R_2)$.\qed
\end{lemma}
\begin{lemma}\label{Lem:2DLMIBasic}\cite{LMI1D}
	Consider the LMI condition $\Phi + P^\tran\Lambda^\tran Q+ Q^\tran\Lambda P<0$,
	where $\Phi\in\fld{R}^{n\times n}$, $P\in\fld{R}^{p\times n}$ and $Q\in\fld{R}^{q\times n}$. There exists a matrix $\Lambda\in\fld{R}^{q\times p}$ satisfying the previous LMI condition if and only if
	$W_p^\tran\Phi W_p<0$ and $W_q^\tran\Phi W_q<0$, where the columns of $W_p$ and $W_q$ are bases of the $\ker P$ and $\ker Q$, respectively.\qed
\end{lemma}

Now consider the 2D system \eqref{Eq:FMII} under the fault free situation and the corresponding  state estimation observer as given by,
\bs
\begin{align}\label{Eq:GeneralObserver}
	\hat{x}(i+1,j+1) = &(A_1+D_{o1}C)\hat{x}(i,j+1) + (A_2+D_{o2}C)\hat{x}(i+1,j) + B_1u(i,j+1)\notag\\ &-D_{o1}y(i,j+1)-D_{o2}y(i+1,j)+ B_2u(i+1,j),\notag\\
	\hat{y}(i,j) = &C\hat{x}(i,j).
\end{align}
\es
It follows readily that the state estimation error dynamics, as defined by $e(i,j)=x(i,j)-\hat{x}(i,j)$, is governed  by,
\begin{equation}\label{Eq:ObserverError}
	\begin{split}
		e(i+1,j+1) = &(A_1+D_{o1}C)e(i,j+1) + (A_2+D_{o2}C)e(i+1,j).
	\end{split}
\end{equation}
The following theorem and corollary provide an LMI-based condition for existence of the state estimation observer gains $D_{o1}$ and $D_{o2}$ such that the error dynamics \eqref{Eq:ObserverError} is asymptotically stable.
\begin{theorem}\label{Thm:ObserverGain2LMIs}
	Consider the 2D system \eqref{Eq:FMII} under the fault free situation.  There exist two maps $D_{o1},D_{o2}:\fld{R}^q\rightarrow\fld{R}^n$ and  two symmetric positive definite matrices $R_1$ and $R_2$ such that the LMI \eqref{Eq:2DLyapunov} is satisfied for $A_1+D_{o1}C$ and $A_2+D_{o2}C$ if and only if $R_1$ and $R_2$ satisfy the LMI condition $W_{c_d}^\tran A_c W_{c_d}<0$, where $W_{c_d}= \diag(W_c,W_c)$ and the columns of $W_c\in\fld{R}^{n\times(n-q)}$ are the basis of $\ker C$.
\end{theorem}
\begin{proof}
	Note that  without loss of any generality, it is assumed that $C$ is full row rank, and $n>q$ that is equivalent to partial state measurement. Let $\Phi = \bbm -(R_1+R_2)^{-1}  &A \\ A^\tran &-R\ebm$. By using the Schur complement lemma, we have $W_{c_d}^\tran A_c W_{c_d}<0$ if and only if,
	\begin{equation}\label{Eq:ConLMI2DTemp}
		\bbm -(R_1+R_2)^{-1}  & A W_{c_d}\\ W_{c_d}^\tran A^\tran  &-W_{c_d}^\tran R W_{c_d}\ebm = \bbm I &0\\ 0 &W_{c_d}^\tran\ebm \Phi \bbm I &0\\ 0 &W_{c_d}\ebm<0.
	\end{equation}
	It follows that
	$\bbm 0_{2n\times n} &I_{2n\times 2n}\ebm \Phi \bbm 0_{n\times 2n} \\I_{2n\times 2n}\ebm = -R<0$ if and only if $R>0$ (or $R_1>0$ and $R_2>0$). By defining $P =\bbm 0_{2q\times n} &C_d\ebm$ and $Q =\bbm I_{n\times n} &0_{n\times 2n}\ebm$, where $C_d = \diag(C,C)$, and using Lemma \ref{Lem:2DLMIBasic} the LMI condition \eqref{Eq:ConLMI2DTemp} is satisfied if and only if there exits a matrix $\Lambda=\bbm D_{o1} &D_{o2}\ebm\in\fld{R}^{n\times 2q}$ such that,
	\bs
	\begin{equation}\label{Eq:ObserverDesignLMI_Gain}
		\begin{split}
			&\Phi + \bbm 0_{n\times 2q} \\ C_d^\tran\ebm \Lambda^\tran \bbm I_{n\times n} &0_{n\times 2n}\ebm +  \bbm I_{n\times n}\\0_{2n\times n}\ebm \Lambda \bbm 0_{2q\times n} &C_d\ebm  = \bbm (R_1+R_2)^{-1} &G\\ G^\tran &R\ebm <0,
		\end{split}
	\end{equation}
	\es
	where $G = \bbm A_1+D_{o1}C &A_2+D_{o2}C\ebm$. Again, by using the Schur complement lemma, we have
	$G^\tran (R_1+R_2)G-R<0$. This completes the proof of the theorem.
\end{proof}
An important corollary to the above theorem and Lemma \ref{Lem:2DLyapunov} can be stated as follows.
\begin{corollary}\label{Col:ObserverDesignLMI}
	Consider the 2D system \eqref{Eq:FMII} under the fault free situation and the state estimation observer \eqref{Eq:GeneralObserver}. If there are two symmetric positive definite matrices $R_1$ and $R_2$ satisfying the LMI condition $W_{c_d}^\tran A_c W_{c_d}<0$, then there exists two maps $D_{o1}$ and $D_{o2}$ such that the error dynamics \eqref{Eq:ObserverError} is asymptotically stable.
\end{corollary}
\begin{proof}
	Follow directly from Theorem \ref{Thm:ObserverGain2LMIs} and Lemma \ref{Lem:2DLyapunov}, and the details are omitted for sake of brevity.
\end{proof}
\begin{remark}\label{Rem:ObserverLMIConstructive}
	Note that by solving the LMI condition $W_{c_d}^\tran A_c W_{c_d}<0$, one can obtain symmetric  positive definite matrices $R_1$ and $R_2$. Hence, the state estimation observer gains $D_{o1}$ and $D_{o2}$ are computed by solving the equation \eqref{Eq:ObserverDesignLMI_Gain} (which is an LMI condition in terms of the gains $D_{o1}$ and $D_{o2}$). Therefore, Corollary \ref{Col:ObserverDesignLMI} not only provides sufficient conditions for existence of a state estimation observer, but also provides an approach for computing the observer gains $D_{o1}$ and $D_{o2}$. \qed
\end{remark}
The results of this section will now be used subsequently in Sections \ref{Sec:UnObservable} and \ref{Sec:FDI} to address the unobservable subspace of the system \eqref{Eq:FMII} as well as to provide sufficient conditions for  solvability of the FDI problem, respectively.

\section{Invariant Subspaces for 2D FMII Models}\label{Sec:InvSpace}
As described earlier, 2D systems can be represented as \infd systems (i.e. the initial condition is a vector of an Inf-D subspace). In this section, we first use the Inf-D representation \eqref{Eq:IDRep} to formally define and construct an unobservable subspace. Next, we define a subspace of  the unobservable subspace  (this we called as an invariant unobservable subspace) of the 2D system \eqref{Eq:FMII} that can be represented as an infinite sum of the same finite dimensional subspaces. Therefore, one can compute the invariant unobservable subspace (that is, the \infd subspace) in a \underline{finite} number of steps. Also, it is shown that the invariant unobservable subspace enjoys an important geometric property that is crucial for solving the FDI problem.

\subsection{Unobservable Subspace}\label{Sec:UnObservable}

As described in the previous section, 2D systems can be represented as Inf-D systems. In this subsection, the Inf-D representation \eqref{Eq:IDRep} is utilized to formally define and construct an unobservable subspace.


The unobservable subspace of the system \eqref{Eq:IDRep} (and consequently of the system \eqref{Eq:FMII}) is defined as,
\begin{equation}\label{Eq:UnobserSpaceGeneral}
	\ssp{N}_g =\bigcap_{k=0}^\infty\ker\op{C}\op{A}^k,
\end{equation}
where $\op{A}$ and $\op{C}$ are defined as in equation \eqref{Eq:IDRep}. Note that we define the above unobservable subspace by following along the steps in \cite{Curtain_1986}, the results in \cite{Zwart_Book} (Chapter I), and the fact that the operator $\op{A}$ in equation \eqref{Eq:IDRep} is bounded (refer to Lemma \ref{Lem:Boundedness}).

One of the main difficulties in geometric analysis of Inf-D systems is the convergence of any developed algorithm that involves computation of certain set of subspaces in a \underline{finite} number of steps. For example, consider the unobservable subspace \eqref{Eq:UnobserSpaceGeneral}. In Fin-D systems, the algorithm for computing the unobservable subspace converges in a finite number of steps \cite{Wonham_Book}. Moreover, one is generally interested in investigating the FMII models in a Fin-D representation \eqref{Eq:FMII}.  Motivated by the above, below two important subspaces of $\ssp{N}_g$ that are denoted by $\ssp{N}_\infty$ and $\ssp{N}_{s,\infty}$ are introduced. The subspaces $\ssp{N}_\infty$ and $\ssp{N}_{s,\infty}$  can be computed in a \underline{finite} number of steps and also allows one to derive necessary and sufficient conditions for solvability of the FDI problem.

Consider the initial condition ${\bf x}(0) = (\cdots,0,x_0,0,\cdots)$ and $u(i,j) = \left\{
\begin{array} {c l}
u_0 &;\; i=j=0\\
0 &;\; \mathrm{otherwise}
\end{array}\right.$, where $u_0\in\fld{R}^m$. One can show that the  state solution of the model \eqref{Eq:FMII} under the fault free situation is given by \cite{FMinBook},
\vspace{-1mm}
\begin{eqnarray}\label{Eq:FMIISol}
	\begin{split}
		x(i,j) = A^{(i,j)}x_0 + A^{(i,j)}_B u_0,
	\end{split}
\end{eqnarray}
where the matrices $A^{(i,j)}$'s and $A^{(i,j)}_B$'s are defined by the following recursive  expressions,
\vspace{-2mm}
\begin{eqnarray}\label{Eq: Adef1}
	\begin{split}
		A^{(i,j)} &= A_1 A^{(i-1,j)} + A_2A^{(i,j-1)},\;\; A^{(i,j)} = 0 \ \ \  \mathrm{if} \ i \ \mathrm{or} \ j < 0,\\
		A^{(i,j)}_B &=A^{(i-1,j)}B_1+A^{(i,j-1)}B_2,\ A^{(0,0)} = I.
	\end{split}
\end{eqnarray}
Based on the solution that is given by equation  \eqref{Eq:FMIISol}, and considering that $u_0=0$, a finite observability matrix (given that its null space is a finite dimensional subspace) can be defined as follows,
\begin{equation}\label{Eq:Observ_M}
	O = \begin{bmatrix}C^\tran, (CA_1)^\tran, (CA_2)^\tran, \cdots, (CA^{(i,j)})^\tran, \cdots\end{bmatrix}^\tran.
\end{equation}
Let $\ssp{N} = \ker\;O = \bigcap_{i,j\geq 0} \big(\ker CA^{(i,j)}\big)$. Since $\dim(\ssp{N})\le n<\infty$, we designate $\ssp{N}$ as the \emph{finite} unobservable subspace of the system \eqref{Eq:FMII}.  Also, recall from the 2D Cayley-Hamilton theorem \cite{FMinBook} that for all $i+j\geq n$, one sets $A^{(i,j)}=\sum_{h+k<n}\zeta_{h,k} A^{(h,k)}$, where $\zeta_{h,k}$'s are real numbers. Therefore, for all $i+j\geq n$, $\bigcap_{h+k<n}\ker CA^{(h,k)}\subseteq\ker CA^{(i,j)}$, and consequently $\ssp{N}$ can be computed in a \underline{finite} number of steps as,
\begin{equation}\label{Eq:UnObsSpaceFinite}
	\ssp{N} = \ker\;O = \bigcap_{i,j\geq 0,\;i+j<n} \big(\ker CA^{(i,j)}\big).
\end{equation}

Now, we consider the following subspace,
\begin{equation}\label{Eq:UnobsSpaceInfinite}
	\ssp{N}_\infty = \sumbanach{\ssp{N}}
\end{equation}
It follows that if ${\bf x}_0 = (\cdots,x_{-1},x_0,x_{1},\cdots)\in\ssp{N}_\infty$, then $x_{i}\in\ssp{N}$ for all $i\in\fld{Z}$, and given the zero input assumption one gets ${\bf y}(k)=0$ for all $k\in\underline{\fld{N}}$ (in equation \eqref{Eq:IDRep}). By considering $\op{A}^k$, where $\op{A}$ is defined as in equation \eqref{Eq:IDRep} and $k\in\fld{N}$, it can be shown that $\ssp{N}_\infty\subseteq\ssp{N}_g$. Also, note that although $\ssp{N}_\infty$ is an \infd subspace, it can be computed in a \underline{finite} number of steps (one only needs to compute $\ssp{N}$). However, as explained in \cite{ACC2013,ACC2014} the invariance property (this is addressed in the next subsection) of $\ssp{N}$ is not lucid. Therefore, in the following a subspace of $\ssp{N}$ is introduced such that it enjoys this geometric property. To define the subspace $\ssp{N}_s$ one needs the following notation.

Let us express $A^{\alpha}$ to denote the sequence of multiplications of $A_1$ and $A_2$,  where $\alpha$ is a multi-index parameter that specifies the sequence of the multiplication.  For example, consider $A^\alpha = A_2A_1A_1A_2A_1$, where we have $\alpha = (2,1,1,2,1)$.
The notation $||\alpha||$ denotes the number of all $A_1$ and $A_2$ that are involved in the corresponding multiplication (for the above example, we have $||\alpha|| = 5$). Now, consider the following subspace (for more details on $\ssp{N}_s$ refer to \cite{ACC2013}),
\begin{equation}
	\ssp{N}_s = \bigcap_{||\alpha||<n}\ker CA^\alpha
\end{equation}

The following lemma shows that the subspace that is used in \cite{ntogramatzidis2012Siam, Malek_3DFDI, Malek_3DFDIConf} as the unobservable (non-observable) subspace is indeed $\ssp{N}_s$.

\begin{lemma}\label{Lm:NsisSub}
	The subspace $\ssp{N}_s$ can be computed in a finite number of steps according to the following algorithm,
	\begin{equation}\label{Eq:NsAlg}
		\begin{split}
			&\ssp{V}_0 = \ker C\;\;	\ssp{V}_k = A_1^{-1}\ssp{V}_{k-1}\cap A_2^{-1}\ssp{V}_{k-1}\cap\ker C.
		\end{split}
	\end{equation}
\end{lemma}
\begin{proof}
	First, note that $\ssp{V}_1 = A_1^{-1}\ker C\cap A_2^{-1}\ker C\cap\ker C$ and $\ssp{V}_2 = \ssp{V}_1 \cap (A_1A_2)^{-1}\ker C \cap (A_2A_1)^{-1}\ker C\cap(A_1^2)^{-1}\ker C \cap (A_2^2)^{-1}\ker C$. In other words, $\ssp{V}_k=\bigcap_{||\alpha||\leq k}(A^{\alpha})^{-1}\ker C$. Note that for every pair of operators $C:\fld{R}^n\rightarrow \fld{R}^q$ and $F:\fld{R}^n\rightarrow \fld{R}^n$, one can show that $\ker CF=F^{-1}\ker C$. Therefore, it follows that $\ssp{V}_n = \ssp{N}_s$. This completes the proof of the lemma.
\end{proof}

Now, we set $\ssp{N}_{s,\infty} = \sumbanach{\ssp{N}_s}$. Note that although $\dim({N}_{s,\infty})=\infty$, one can compute it in a finite number of step (by computing $\ssp{N}_s$).

\subsection{$A_{1,2}$-Invariant Subspaces}\label{Sec:A12Inv}
As stated in the Subsection \ref{Sec:IDRep}, the 2D system \eqref{Eq:FMII} can be represented as an \infd system \eqref{Eq:IDRep}. In order to formulate the corresponding Inf-D invariant subspaces one needs the next two definitions.
\begin{definition}\cite{Zwart_Book}\label{Def:AID_Inv}
	Consider the \infd system \eqref{Eq:IDRep}, where the operator $\op{A}$ is bounded (according to Lemma \ref{Lem:Boundedness}). The closed subspace $\ssp{V}_\infty\in\op{X} = \sum\fld{R}^n$ is called $\op{A}$-invariant if $\op{A}\ssp{V}_\infty\subseteq\ssp{V}_\infty$.\qed
\end{definition}

\begin{definition}\cite{conte1988GeometryConf}\label{Def:A12_Inv}
	The subspace $\ssp{V}\subset\fld{R}^{n}$  is said to be an {\it $A_{1,2}$-invariant subspace} for the 2D system \eqref{Eq:FMII} if
	$A_1\ssp{V} + A_2\ssp{V}\subseteq \ssp{V}$,
	where $A_1$ and $A_2$ are defined as in equation \eqref{Eq:FMII}.\qed
\end{definition}
Note that $\ssp{V}$ is $A_{1,2}$-invariant if and only if it is invariant with respect to $A_1$ \underline{and} $A_2$ (i.e. $A_1\ssp{V}\subseteq\ssp{V}$ and $A_2\ssp{V}\subseteq\ssp{V}$).
The following theorem provides the connection between the Definitions \ref{Def:AID_Inv} and \ref{Def:A12_Inv}.
\begin{theorem}\label{Thm:AID_A12Inv}
	Consider the 2D system \eqref{Eq:FMII} and the \infd system \eqref{Eq:IDRep}. Let $\ssp{V}_\infty = \sum \ssp{V}$, where $\ssp{V}\subseteq\fld{R}^n$. The subspace $\ssp{V}_\infty$ is $\op{A}$-invariant if and only if $\ssp{V}$ is $A_{1,2}$-invariant.
\end{theorem}
\begin{proof}
	First, note that every ${\bf x}\in\ssp{V}_\infty$ can be expressed as ${\bf x}=\sum_{k=-\infty}^\infty {\bf x}_k^k$, where ${\bf x}_k^k = \bbm \cdots,0,0,x_k^\tran,0,0\cdots\ebm^\tran\in\ssp{V}_\infty$ and $x_k\in\ssp{V}$. Therefore, one only needs to show the result for ${\bf x}_k^k$.\\
	({\bf If part}): Assume $\ssp{V}$ is $A_{1,2}$-invariant. Consider the Inf-D vector ${\bf x}_k^k$. It follows that $\op{A}{\bf x}_k^k= [\cdots, 0, (A_2x_k)^\tran$ $, (A_1x_k)^\tran,$ $ 0, 0\cdots]^\tran$. Since $\ssp{V}$ is $A_{1,2}$-invariant, it follows that $\op{A}{\bf x}_k^k\in\ssp{V}_\infty$. \\
	({\bf Only if part}): Let $\op{A}\ssp{V}_\infty\subseteq\ssp{V}_\infty$ and ${\bf x}_0^0\in\ssp{V}_\infty$. Consequently, $x_0\in\ssp{V}$. Since $\op{A}{\bf x}_0^0= [ \cdots,0 , $ $ (A_2x_0))^\tran, (A_1x_0)^\tran, 0, 0, \cdots]^\tran\in\ssp{V}_\infty$, it follows that $A_1x_0\in\ssp{V}$ and $A_2x_0\in\ssp{V}$, and consequently $\ssp{V}$ is $A_{1,2}$-invariant. This completes the proof of the theorem.
\end{proof}

Consider the subspaces $\ssp{V}\subseteq\fld{R}^n$ and $\ssp{C}\subseteq\fld{R}^n$. If $\ssp{V}$ is the largest $A_{1,2}$-invariant subspace that is contained in $\ssp{C}$, we denote $\ssp{V} = <\ssp{C}|A_{1,2}>$. We have shown in \cite{ACC2013} that $\ssp{N}_s\subseteq\ssp{N}$, and it is the largest $A_{1,2}$-invariant subspace that is contained in $\ker C$. Therefore,  one can write $\ssp{N}_s = <\ker C|A_{1,2}>$. Since $\ssp{N}_s$ is $A_{1,2}$-invariant, by Theorem \ref{Thm:AID_A12Inv}, $\ssp{N}_{s,\infty}$ is $\op{A}$-invariant. Therefore, if ${\bf x}(0) = (\cdots,x_{-1},x_0,x_1,\cdots)\in\ssp{N}_{s,\infty}$ (that is, $x_i\in\ssp{N}_s$ for all $i\in\fld{Z}$) and zero input, ${\bf x}(k)\in\ssp{N}_{s,\infty}$ for all $i\in\fld{Z}$ and ${\bf y}(k)=0$ for all $k\in\underline{\fld{N}}$ (in equation \eqref{Eq:IDRep}). We designate  $\ssp{N}_{s,\infty}$ as the invariant unobservable subspace.
\begin{remark}\label{Rem:InitCond}
	As stated in Subsection \ref{Sec:IDRep}, there are two different types of initial condition formulations. In this paper, we use the first formulation that is compatible with the Inf-D system \eqref{Eq:IDRep}. Recall that the second formulation is expressed as $x(i,0)=h_1(i)$ and $x(0,j)=h_2(j)$, where $i,j\in\fld{N}$.  Now, let $x(i,0)\in\ssp{N}_s$ and $x(0,j)\in\ssp{N}_s$. The $A_{1,2}$-invariance property of $\ssp{N}_s$ verifies that $y(i,j)=0$. In other words, $\ssp{N}_{s,\infty}$ is also the invariant unobservable subspace of system \eqref{Eq:FMII} with the second formulation of the initial conditions. Therefore, without loss of any generality, one can apply our proposed approach to \underline{both} initial condition formulations as provided in Section \ref{Sec:Dis_2D}. Moreover, $\ssp{N}_{s,\infty}$ is the largest $\op{A}$-invariant in the form $\ssp{N}_{s,\infty}=\sumbanach{\ssp{N}_s}$ that is contained in $\ker\op{C}$, where $\op{C}$ is defined in \eqref{Eq:IDRep}.
\end{remark}

\subsection{Conditioned Invariant Subspaces}

Another important subspace in the geometric FDI toolbox is the conditioned invariant (i.e., the $(C,A_{1,2})$-invariant) subspace that is defined next. This definition is an extension of the one that has appeared and presented in \cite{conte1988GeometryConf} and \cite{conte1988GeometryArticle}.
\begin{definition}\label{Def:CondInv}
	The subspace $\ssp{W}_\infty=\sum \ssp{W}$ (where $\ssp{W}\subseteq\fld{R}^{n}$) is said to be the conditioned invariant subspace for the 2D system \eqref{Eq:FMII} if there exist two output injection maps $D_1,D_2: \fld{R}^{q}\rightarrow \fld{R}^{n}$ such that $ (A_1+D_1C)\ssp{W}\subseteq\ssp{W}$ and
	$(A_2+D_2C)\ssp{W}\subseteq\ssp{W}$. In other words, $\ssp{W}$ is  $[A+DC]_{1,2}$-invariant (i.e., invariant with respect to $A_1+D_1C$ and $A_2+D_2C$). We designate $\ssp{W}$ as the finite conditioned invariant subspace (since $\dim(\ssp{W})<\infty$) of the 2D system \eqref{Eq:FMII}.\qed
\end{definition}

Similar to 1D systems, one can now state the following result.
\begin{lemma}\label{Lem:CA_Inv}
	The following statements are equivalents.
	\begin{enumerate}
		\renewcommand{\labelenumi}{(\roman{enumi})}
		\item The subspace $\ssp{W}_\infty$ is conditioned invariant.
		\item  $A_1(\ssp{W}\cap\ker C)+A_2(\ssp{W}\cap\ker C)\subseteq\ssp{W}$.
		\item $\op{A}(\ssp{W}_\infty\cap\ker\op{C})\subseteq\ssp{W}_\infty$.
	\end{enumerate}
	where $\ssp{W}_\infty = \sum \ssp{W}$.
\end{lemma}
\begin{proof}
	$(i)\Leftrightarrow(ii)$ and $(i)\Leftrightarrow(iii)$: By definition, there exists two maps $D_1$ and $D_2$ such that $\ssp{W}$ is $[A+DC]_{1,2}$-invariant. By utilizing Theorem \ref{Thm:AID_A12Inv} $\ssp{W}_\infty$ is $\op{A}_d$-invariant, where
	\begin{equation}
	\op{A}_d = \bbm \ &\ddots &\ddots &\cdots & &\cdots\\
	\cdots &0 &A_1+D_1C &A_2+D_2C &0 &\cdots\\
	\cdots &0 &0 &A_1+D_1C &A_2+D_2C &\cdots\\
	\cdots & &\cdots & &\ddots &\ddots \ebm
	\end{equation}
	By following along the same lines as in Lemma \ref{Lem:Boundedness}, one can show that $\op{A}_d$ is bounded. Consequently, the result of 1D system is also valid for the Inf-D system \eqref{Eq:IDRep}. Hence, we have $\op{A}(\ssp{W}_\infty\cap\ker\op{C})\subseteq\ssp{W}_\infty$ (that shows $(i)\Leftrightarrow(iii)$). By considering the structure of $\op{A}$ and $\op{C}$ it follows that $A_1(\ssp{W}\cap\ker C)+A_2(\ssp{W}\cap\ker C)\subseteq\ssp{W}$.\\
	$(iii)\Leftrightarrow(i)$: Since $\op{A}$ is bounded, the domain of $\op{A}$ is equal to $\op{X}=\sum \fld{R}^n$, and therefore, the result of 1D system is also valid for the Inf-D system \eqref{Eq:IDRep}. Therefore, there exists a bounded operator $\op{D}$ such that $\ssp{W}_\infty$ is $\op{A+DC}$-invariant. By considering the structure of $\ssp{W}_\infty$ and $\ker\op{C}$, it is easy to show that one solution for $D$ is given by
	\begin{equation}
	\op{D} = \bbm \ &\ddots &\ddots &\cdots & &\cdots\\
	\cdots &0 &D_1 &D_2 &0 &\cdots\\
	\cdots &0 &0 &D_1 &D_2 &\cdots\\
	\cdots & &\cdots & &\ddots &\ddots \ebm
	\end{equation}
	Hence, by using Theorem \ref{Thm:AID_A12Inv}, the subspace $\ssp{W}$ is $[A+DC]_{1,2}$-invariant and consequently $\ssp{W}_\infty$ is a conditioned invariant subspace. This completes the proof of the lemma.
\end{proof}

In the geometric FDI approach, one is interested in conditioned invariant subspaces that are containing a given subspace \cite{Massoumnia1989}. Let us define all the conditioned invariant subspaces containing a subspace $\ssp{L}_\infty = \sum \ssp{L}$ ($\ssp{L}\subseteq\fld{R}^n$) as  $\mathfrak{Q}(\ssp{L}) = \{\ssp{W}_\infty| \ \exists\; D_1\;\mathrm{and}\; D_2 \ ;\; (A_i+D_iC)\ssp{W}\subseteq\ssp{W}\; \mathrm{and}\; \ssp{W}\supseteq\ssp{L} \}$.  It can be shown that for a given subspace $\ssp{L}_\infty$ (or $\ssp{L}$), the set $\mathfrak{Q}(\ssp{L})$ is closed under intersection, and hence the set $\mathfrak{Q}(\ssp{L})$ has a minimal member as $\ssp{W}_\infty^{*}(\ssp{L})$. The minimal conditioned invariant subspace containing a given subspace $\ssp{L_\infty}=\sum \ssp{L}$ (that is, $\ssp{W}_\infty^{*}(\ssp{L})$) is obtained by invoking the following non-decreasing algorithm that is provided below,
\begin{eqnarray}\label{Eq:CIAlg}
	\begin{split}
		\ssp{W}^0 &= \ssp{L},\\
		\ssp{W}^k &= \ssp{L} + \big(A_1(\ssp{W}^{k-1}\cap\ker C)+ A_2(\ssp{W}^{k-1}\cap\ker C)\big),
	\end{split}
\end{eqnarray}
and $\ssp{W}_\infty^* = \sum\ssp{W}^*$, where $\ssp{W}^*=\ssp{W}^{k_0}, k_0\leq n$. 

Note that the above algorithm converges in a \underline{finite} number of steps. Also, let $\ssp{W}$ be a finite conditioned invariant subspace.  The set of all maps $D = [D_1, D_2]$ such that $\ssp{W}$ is $[A+DC]_{1,2}$-invariant is designated by $\underline{D}(\ssp{W})$.
\subsection{Unobservability Subspace}
The unobservability subspace \cite{Mass_Thesis, DrMeskin_Delay} is the cornerstone of geometric FDI approach in 1D systems. The following definition generalizes and extends this concept to the FMII 2D models. The Fin-D notion of this definition, for the first time in the literature, was introduced and utilized in \cite{ACC2013} for the FDI problem of 2D Roesser systems. 
\begin{definition}\label{Def:Unobser}
	The subspace $\op{S}_\infty$ is said to be an unobservability subspace for the 2D system \eqref{Eq:FMII} if there are three maps $D_1,D_2,H$ such that
	$\op{S} = <\ker HC|[A+DC]_{1,2}>$ and $\op{S}_\infty =\sum\op{S}$. We designate $\op{S}$ as the finite unobservability subspace of the 2D system \eqref{Eq:FMII}.
	\qed
\end{definition}
Note that $\op{S}_\infty$ is also conditioned invariant subspace and a generic unobservable subspace of the system ($HC$, $A_1+D_1C$, $A_2+D_2C$).
For accomplishing the gaol of the FDI task, one first computes an unobservability subspace and then obtains the matrix $H$ \cite{Mass_Thesis}. Therefore, it is necessary to compute the unobservability subspace without having any knowledge of $H$. By following along the same lines as those used in \cite{ACC2013} (Section IV.D), one can show that the limit of the following algorithm is the smallest unobservability subspace $\op{S}^*$ (and consequently $\op{S}_\infty^*$) that contains a given subspace $\ssp{L}$, that is,
\begin{equation}\label{Eq:Unob_MainAlg}
	\begin{split}
		\ssp{Z}^0 &= \fld{R}^{n}\; \text{and}\;
		\ssp{Z}^{k} = \ssp{W}^{*} + \big( A_1^{-1}\ssp{Z}^{k-1} \cap A_2^{-1}\ssp{Z}^{k-1} \cap \ker C\big),
	\end{split}
\end{equation}
where $\ssp{W}^{*}$ is the minimal finite conditioned invariant subspace containing $\ssp{L}$, $\op{S}_\infty^*=\sum \op{S}^*$ and $\op{S}^* = \ssp{Z}^n$.

Finally, we set $\op{S}_\infty^*=\sum \op{S}^*$, where $\op{S}^* = \ssp{Z}^n$. Finally, it should be noted that since $\ssp{W}^*+\ker C = \op{S}^*+\ker C$ (this follows by applying $\ssp{Z}^{n+1} = \ssp{Z}^n = \op{S}^*$ in the algorithm \eqref{Eq:Unob_MainAlg}), one obtains $\underline{D}(\ssp{W}^*)\subseteq\underline{D}(\op{S}^*)$.

In this section, we first defined the invariance property of the 2D system \eqref{Eq:FMII} that are Inf-D subspace of the \infd system \eqref{Eq:IDRep}. Next, an invariance unobservable subspace (that is generically equivalent to an unobservable subspace) $\ssp{N}_{s,\infty}=\sum\ssp{N}_s$ was introduced. Moreover, the conditioned and unobservability subspaces that are crucial in determining the solution to the FDI problem have been introduced. By utilizing the above results necessary and sufficient conditions for solvability of the FDI problem are subsequently derived and provided.

\section{Necessary and Sufficient Conditions for Solvability of the FDI problem}\label{Sec:FDI}
In this section, we first present  sufficient conditions for detectability and isolability of faults. Next, by employing three different filters (namely, ordinary, delayed deadbeat, and LMI-based filters) sufficient conditions for solvability of the FDI problem are presented.

Consider the faulty FMII model \eqref{Eq:FMII} (i.e., the system is subjected to two faults $f_1$ and $f_2$) and the detection filter \eqref{Eq:Filter} designed to detect and isolate the fault $f_1$. By augmenting the detection filter dynamics \eqref{Eq:Filter} with the faulty 2D model \eqref{Eq:FMII}, one obtains,
\bs
\begin{equation}\label{Eq:AugSys}
	\begin{split}
		x_e(&i+1,j+1) = A_1^ex_e(i,j+1) + A_2^ex_e(i+1,j) + B_1^e u(i,j+1)+ B_2^e u(i+1,j)+ L_{2,e}^1f_2(i,j+1) \\&+ L_{2,e}^2f_2(i+1,j)+L_{1,e}^1f_1(i,j+1)+L_{1,e}^2f_1(i+1,j),\\
		r_1(&i,j) =  C_ex_e(i,j),
	\end{split}
\end{equation}
\es
where $x_e = \bbm x^\tran &\omega^\tran\ebm^\tran\in \ssp{X}^e = \fld{R}^n\oplus\fld{R}^o$ ($o$ refers to the dimension of $\omega$), $C_e = \bbm H_1C &M_1\ebm$ and,
\bs
\begin{equation}\label{Eq:AugSysPar}
	\begin{split}
		A_1^e &= \bbm A_1 &0\\E_1C &F_1\ebm, A_2^e = \bbm A_2 &0\\ E_2C &F_2\ebm, B_1^e = \bbm B_1 \\K_1\ebm ,\\
		B_2^e &= \bbm B_2 \\ K_2 \ebm , L_{i,e}^j = \bbm (L_i^j)^\tran &0\ebm^\tran, \; i,j\in\{1,2\}.
	\end{split}
\end{equation}
\es
In this section, by assuming that the unobservable subspace of the above augmented system is $A^e_{1,2}$-invariant, it is shown that the sufficient condition is also necessary for solvability of the FDI problem. Moreover, an analytical  comparison between our proposed approach and the method developed in \cite{Malek_3DFDI} is also provided to highlight the strength and capabilities of our proposed methodology when compared to the literature.
\subsection{Main Results}\label{Sec:FDIMainResult}
The following lemma provides an important property for the invariant unobservable subspace $\ssp{N}_{s,\infty}^e$ (and $\ssp{N}_s^e$) that is associated with the system \eqref{Eq:AugSys}.
\begin{lemma}\label{Lem:UnobsTot2Small}
	Consider the 2D system \eqref{Eq:AugSys} and its invariant unobservable subspace $\ssp{N}_{s,\infty}^e$. Then $\op{Q}^{-1}\ssp{N}^e_{s,\infty}$ is an unobservability subspace of the 2D system \eqref{Eq:FMII}, with $\op{Q} = \diag(\cdots,Q,Q,\cdots)$ and $Q$ representing the embedding operator into $\ssp{X}^e$ (i.e., $Q:\fld{R}^n\rightarrow\ssp{X}^e$ and $Qx = [x^\tran, 0]^\tran$).
\end{lemma}
\begin{proof}
	First, recall that $\ssp{N}_{s,\infty}^e = \sum \ssp{N}_s^e$.
	Note that, $Q^{-1}\ssp{N}_s^e= \ssp{S}=\{x | \left[\bsm x\\0\esm\right] \in \ssp{N}_s^e\}$, and 
	assume that $\left[\bsm x\\0\esm\right]\in \ssp{N}_s^e$. According to the fact that $\ssp{N}_s^e$ is $A^e_{1,2}$-invariant \cite{ACC2013}, we have $A_1^e\left[\bsm x\\0\esm\right]\in\ssp{N}_s^e$, and if $x\in\ker C$ then $A_1x\in \ssp{S}$. It follows that $A_1( \ssp{S}\cap\ker C)\subseteq \ssp{S}$. By following along the same lines as those above one can show that $A_2(\ssp{S}\cap\ker C)\subseteq \ssp{S}$. Therefore, $\ssp{S}_\infty=\op{Q}^{-1}\ssp{N}^e_{s,\infty}=\sum \ssp{S}$ is a conditioned invariant subspace.
	Moreover, given that $\ssp{N}_s^e$ is contained in $\ker C_e$, we have $\ssp{S}\subseteq\ker H_1C$. Therefore, the subspace $\ssp{S}$ is a finite conditioned invariant subspace contained in $\ker {H_1C}$. Since $\ssp{N}_s^e$ is the largest $A_{1,2}^e$-invariant subspace in $\ker C_e$, it follows that $\ssp{S}$ is the largest $[A+DC]_{1,2}$ invariant subspace in $\ker H_1C$ (i.e., $\ssp{S}$ is a finite unobservability subspace of the 2D system \eqref{Eq:FMII}). In other words, $\ssp{S}$ is a finite unobservability subspace of the system \eqref{Eq:FMII}. This completes the proof of the lemma.
\end{proof}
The following theorem provides a \underline{single} necessary and  sufficient condition for detectability and isolability of faults (i.e., the existence  of a subsystem such that it is decoupled from all faults but one - refer to  Subsection \ref{Sec:FDIProb} for more details).
\begin{theorem}\label{Thm:NecSuffFDI}
	Consider the 2D system \eqref{Eq:FMII} that is subject to two faults $f_1$ and $f_2$. The fault $f_1$  is detectable and isolable (in the sense of Remark \ref{Rem:ProbName}) if and only if the following condition is satisfied,
	\begin{equation}\label{Eq:NecSufCon}
	\begin{split}
	(\ssp{L}_1^{1}+\ssp{L}_1^{2})\not\subseteq\op{S}_1^*
	\end{split}
	\end{equation}
	where  $\op{S}_1^*$ is the smallest finite uobservability subspace (refer to Definition \ref{Def:Unobser}) of the 2D system \eqref{Eq:FMII} containing $\ssp{L}_2^1+\ssp{L}_2^2$ (this represents the limit of the algorithm \eqref{Eq:Unob_MainAlg}, where one sets $\ssp{L}=\ssp{L}_2^1+\ssp{L}_2^2$ in the algorithm \eqref{Eq:CIAlg}).\qed
\end{theorem}
\begin{proof}
	({\bf If part}): Let $H_1$, $D_1^1$ and $D_2^1$ be the corresponding operators to $\op{S}_1^*$ (i.e. $\op{S}_1^* = <\ker H_1C|A+DC>$). Also, consider the following detection filter, governed according to 
		\begin{align}\label{Eq:Filter_F1}
		\omega_1(i+1,j+1) = & F_1\omega_1(i,j+1)+F_2\omega_1(i+1,j)\notag\\&+P_1B_1u(i,j+1)+P_1B_2(i+1,j),\notag\\
		r_1(i,j) = &M_1\omega_1(i+1,j)-H_1y(i,j),
		\end{align}
		where $P_1: \mathbb{R}^{n}\rightarrow \mathbb{R}^{n}/\ \op{S}_1^*$ is the canonical projection of $\fld{R}^n$ on $\fld{R}^n/\op{S}_1^*$, $F_k = A_p^k+D_{o,k}M_1$, $k=1,2$ where $A_1^pP_1 = P_1(A_1+D_1C)$ and $A_2^pP_1 = P_1(A_2+D_2C)$, $D_{o,k}$ are filter gains  and $M_1$ is the unique solution to $M_1P_1 = H_1C$.
		Now, by defining $e(i,j) = P_1x(i,j)-\omega(i,j)$, one obtains
		\begin{align}\label{Eq:Res}
		e(i+1,j+1) &= F_1e(i,j+1)+F_2e(i+1,j)\notag\\&+P_1f_1(i,j)+P_1f_1(i,j),\notag\\
		r_1(i,j) &= M_1e(i,j).
		\end{align}
		Note that $P_1L_2^1 = P_1L_2^2 = 0$ and by the condition \eqref{Eq:NecSufCon}, at least for $k=1$ and/or $k=2$, we have $P_1L_1^k\neq0$. Therefore, the residual signal $r_1$ is decoupled from $f_2$ and is sensitive to $f_1$. In other words, $f_1$ is detectable and isolable.\\
		{(\bf Only if part)}: We show the necessary part by contradiction. Let $f_1$ be detectable and isolable, and $(\ssp{L}_1^{1}+\ssp{L}_1^{2})\subseteq\op{S}_1^*$. Hence, the residual signal $r_1$ is decoupled from $f_2$ and sensitive to $f_1$. However, from Lemma \ref{Lem:UnobsTot2Small} one obtains $Q(\ssp{L}_1^{1}+\ssp{L}_1^{2})\subseteq\ssp{N}^e$. Consequently, the fault $f_1$ is not detectable by using the residual $r_1$ in \eqref{Eq:Res}, which is in contradiction with the solvability of the FDI problem. 
\end{proof}

	As stated in Remark \ref{Rem:ProbName}, the FDI problem has two main steps. Theorem \ref{Thm:NecSuffFDI} provides the necessary and sufficient condition for the first step (that is, detectability and isolability of $f_1$). Therefore,  condition \eqref{Eq:NecSufCon} is also necessary for solvability of the FDI problem. For the second step, we need to determine $\op{D}_{o,k}$, and consequently $F_k$, $k=1,2$ in the filter \eqref{Eq:Filter_F1} such that the residual dynamics \eqref{Eq:Res} is asymptotically stable. 
	In what follows, we provide the conditions that are based on three different observers (namely, the ordinary, the delayed deadbeat, and the LMI-based filters). We first need to present the following theorem. However, one first needs the following notation. Consider the PBH matrix \eqref{Eq:PBHMatrix}. The matrix $N(z_1,z_2) = [ N_{1}(z_1,z_2),$ $\;N_{2}(z_1,z_2)]$ is the minimal left annihilator of the PBH matrix if $N(z_1,z_2)PBH(z_1,z_2) = 0$, $N(z_1,z_2)$ is full row-rank, and any other left annihilator $N_l(z_1,z_2)$ can be expressed as $N_l = N_3N$, where $N_3$ is a polynomial matrix of $z_1$ and $z_2$.
\begin{theorem}\cite{BisiaccoLetter}\label{Thm:BisiacoLetter}
	Consider the 2D system \eqref{Eq:FMII}. The FDI problem is solvable (by using the approach that is proposed in \cite{BisiaccoLetter}) if and only if,
	\begin{equation}\label{Eq:BisiaccoLetterCon}
		\rank(N(z_1,z_2)]\bbm z_1L^1+z_2L^2\\
		0_{q\times  p}\ebm) = p\;,\;\;\forall z_1,z_2\in\fld{C}-\{0\}
	\end{equation}
	where $p$ is the number of faults, $N(z_1,z_2)=\bbm N_1,\;N_2\ebm$ denotes the minimal left annihilator of the PBH matrix  \eqref{Eq:PBHMatrix}, and $L^i = \bbm L_1^k,\;\cdots,\;L_p^k\ebm$ for $k=1,2$.
\end{theorem}

The following corollary provides the sufficient conditions for solvability of the FDI problem by utilizing  our proposed method as well as a (delayed) deadbeat filter.
\begin{corollary}\label{Col:FDI_DelayedDeadBeat_Suff}
	Consider the 2D system \eqref{Eq:FMII} that is subject to two faults $f_1$ and $f_2$. Provided condition \eqref{Eq:NecSufCon} is satisfied and  condition of Theorem \ref{Thm:BisiacoLetter} is satisfied for the quotient subsystem ($H_1C$,$A_1^p$,$A_2^p$), where $H_1,A_1^p$ and $A_2^p$ are defined in \eqref{Eq:Filter_F1}, then  the FDI problem is solvable. 
\end{corollary}
\begin{proof}
	 By invoking Theorem \ref{Thm:NecSuffFDI},  the residual signal $r_1$ in \eqref{Eq:Filter_F1} is decoupled from all faults but $f_1$. Now, given that the condition \eqref{Eq:BisiaccoLetterCon} is satisfied for the subsystem ($H_1C$,$A_1^p$,$A_2^p$), by determining the observer gains $\op{D}_{o,1}$ and $\op{D}_{o,2}$, one can  design a delayed deadbeat observer such that $r_1(i-n_1,j-n_2)=f_1(i,j)$ (refer to \cite[Theorem 1]{BisiaccoLetter}). Hence, the FDI problem is solvable. This complete the proof of the corollary.
\end{proof}

The following corollary provides a sufficient condition for solvability of the FDI problem by utilizing  an {\it ordinary (without delay) deadbeat} detection observer.

\begin{corollary}\label{Col:FDI_DeadBeat_Suff}
	Consider the 2D system \eqref{Eq:FMII} that is subject to two faults $f_1$ and $f_2$. Let the condition \eqref{Eq:NecSufCon} be satisfied and the PBH matrix of the quotient subsystem ($H_1C$,$A_1^p$,$A_2^p$) be right zero prime, where $H_1,A_1^p$ and $A_2^p$ are defined in \eqref{Eq:Filter_F1}. Consequently, the FDI problem is solvable by using a deadbeat (without a delay) observer. 
\end{corollary}
\begin{proof}
	 By invoking Theorem \ref{Thm:NecSuffFDI},  the residual signal in \eqref{Eq:Res} is only affected by $f_1$. Therefore, the detection and isolation of $f_1$ is reduced to that of designing $\op{D}_{o,1}$ and $\op{D}_{o,2}$ gains.  Since the PBH matrix of the subsystem ($H_1C$,$A_1^p$,$A_2^p$) is right zero prime, by invoking Theorem \ref{Thm:DeadbeatObsExistance} one can determine the gains such that the error dynamics \eqref{Eq:Res} is asymptotically stable and the FDI problem is solvable. This complete the proof of the corollary.
\end{proof}

As shown in Theorem \ref{Thm:NecSuffFDI}, under certain conditions one can obtain a residual signal that is decoupled from all faults but one. Therefore, design of a residual generator to detect and isolate the fault $f_1$ in the 2D system \eqref{Eq:FMII} is reduced to that of detecting this fault in the system \eqref{Eq:Filter_F1} by using an observer. The Corollaries \ref{Col:FDI_DelayedDeadBeat_Suff} and \ref{Col:FDI_DeadBeat_Suff} provide sufficient conditions for solvability of the FDI problem by using \emph{delayed} and {\it ordinary deadbeat filters}, respectively. However, as pointed out in \cite{Bisiacco_Obs},  design of an observer for FMII models is based on polynomial matrices. This method is unfortunately not always numerically or analytically straightforward to develop and therefore, in this work we also develop another set of sufficient conditions for solvability of the FDI problem by using a 2D Luenberger observer.

The next corollary provides sufficient conditions for solvability of the FDI problem by using a 2D Luenberger observer.
\begin{corollary}\label{Col:FDI_LMI_Suff}
	Consider  the 2D model \eqref{Eq:FMII}, where the condition \eqref{Eq:NecSufCon} is satisfied. The FDI problem is solvable if there exist two symmetric  positive definite matrices $R_1$ and $R_2$ such that,
	\begin{equation}\label{Eq:LMICon4ObsFactorOurSys}
		W_m^\tran (\bbm R_1 &0\\ 0 &R_2\ebm - \bbm (A_1^p)^\tran\\(A_2^p)^\tran\ebm (R_1+R_2)[A_1^p,A_2^p])W_m<0,
	\end{equation}
	where $A_1^p$ and $A_2^p$ are defined in equation \eqref{Eq:Filter_F1} and the columns of $W_m$ are the basis of $\ker M$.
\end{corollary}
\begin{proof}
	
	By invoking Theorem \ref{Thm:NecSuffFDI}, the fault $f_1$ is detectable and isolable, and consequently the residual signal $r_1$ in \eqref{Eq:Filter_F1} exists (and is decoupled from all faults but $f_1$). By considering the LMI condition \eqref{Eq:LMICon4ObsFactorOurSys} and invoking results from Corollary \ref{Col:ObserverDesignLMI}, the error dynamics \eqref{Eq:Res} is asymptotically stable. If $f_1(\cdot,\cdot)=0$, the residual signal $r_1(\cdot,\cdot)$ converges to zero as $i+j\rightarrow\infty$. Otherwise, the residual has a value that is different from zero. Therefore, the condition \eqref{FDI_InObs1} is satisfied and the FDI problem is solvable. By following along the same procedure as those above one can also design another state estimation observer to detect and isolate the fault $f_2(\cdot,\cdot)$. Therefore, this completes the proof of the corollary.
\end{proof}
\begin{remark}\label{Rem:InfFactor}
	It is worth noting that one can directly work with $\ssp{N}_g$ (as defined in \eqref{Eq:UnobserSpaceGeneral}) and derive necessary and sufficient conditions by following along the same steps as those that have been proposed in \cite{ECC2014}. However, there are two main drawbacks associated with this approach that are as follows:
	\begin{enumerate}
		\item The invariant subspaces are \underline{not} necessarily computed in a finite number of steps.
		\item By factoring out $\ssp{N}_g$, the resulting subsystem does not  necessarily have a Fin-D representation. For more clarification on 2D realization, refer to the work in \cite{FMMinimalRealization}.
	\end{enumerate}
\end{remark}

To summarize, Theorem \ref{Thm:NecSuffFDI} provides a single necessary sufficient condition for  detectability and isolability  of the faults (refer to Subsection \ref{Sec:FDI}) that is also necessary for solvability of the FDI problem. Three sets of sufficient conditions for solvability of the FDI problem by utilizing (i) a (delayed) deadbeat observer, (ii) an ordinary deadbeat observer, and (ii) a 2D Luenberger observer are derived in Corollaries \ref{Col:FDI_DelayedDeadBeat_Suff},  \ref{Col:FDI_DeadBeat_Suff} and \ref{Col:FDI_LMI_Suff}, respectively.

Table \ref{Tab:Suff_Cond} summarizes the main results that are developed and presented in this subsection.
\begin{table}[h]
	\caption{Pseudo-algorithm to detect and isolate the fault $f_i$ in the 2D system \eqref{Eq:FMII}.}
	\centering
	\begin{tabularx}{1\columnwidth}{|X|}
		\hline
		\hline
		\begin{enumerate}
			\item {\tt Compute the minimal conditioned invariant subspace $\ssp{W}^*$ containing all $\ssp{L}_j^1+\ssp{L}_j^2$ subspaces such that $j\neq i$ (by invoking the algorithm \eqref{Eq:CIAlg}, where $\ssp{L}=\sum_{j\neq i} (\ssp{L}_j^1+\ssp{L}_j^2)$).}
			\item {\tt Compute the unobservability subspace $\op{S}_i^*$ containing $\sum_{j\neq i} (\ssp{L}_j^1+\ssp{L}_j^2)$ (by using the algorithm \eqref{Eq:Unob_MainAlg}).}
			\item {\tt Compute the operators $D_1$, $D_2$ such that $\ssp{W}^*$ is the minimal conditioned invariant subspace of the 2D system \eqref{Eq:FMII}.}
			\item {\tt Find the operator $H$ such that $\ker HC = \op{S}_i^*+\ker C$.}
			\item {\tt If $\op{S}_i^*\cap\ssp{L}_i^1\neq 0$ or $\op{S}_i^*\cap\ssp{L}_i^2\neq 0$, then the fault  $f_i$ is detectable and isolable (refer to Remark \ref{Rem:ProbName} and Theorem \ref{Thm:NecSuffFDI}) and }
			\begin{itemize}
				\item {\tt If the conditions of the Corollary \ref{Col:FDI_DelayedDeadBeat_Suff} are satisfied, the FDI problem is solvable by using a (delayed) deadbeat filter.}
				\item {\tt If the conditions of the Corollary \ref{Col:FDI_DeadBeat_Suff} are satisfied, the FDI objective can be accomplished by using a deadbeat (without a delay) observer.}
				\item {\tt If the conditions of Corollary \ref{Col:FDI_LMI_Suff} are satisfied, there exist an LMI-based observer for detection and isolation of the fault $f_i$.}
			\end{itemize}
			\item {\tt The  output norm of any of the above detection filters is the residual that satisfies the condition \eqref{Eq:FDIP}.}
		\end{enumerate}\\
		\hline
	\end{tabularx}
	\label{Tab:Suff_Cond}
\end{table}
\subsection{Comparisons with the Other Available Approaches in the Literature}\label{Sec:FDIComparison}
In this subsection, our proposed approach is now compared and evaluated with respect to the existing geometric methods in the literature \cite{ Malek_3DFDI,Malek_3DFDIConf}. We first show that if the FDI problem is solvable by using the  approach in the above literature, our approach can also detect and isolate the faults. Furthermore, we provide a numerical example where it is shown that our approach is capable of detecting and isolating a fault, however, the necessary conditions provided in the algebraic methods in \cite{BisiaccoMultiDim,BisiaccoLetter} as well as \cite{Malek_3DFDI} are not satisfied.

The equivalent 2D version of the necessary condition (as derived in \cite{Malek_3DFDI}, Theorems 2 and 3) to detect and isolate the fault $f_1$ can be summarized as follows: The fault $f_1$ in the 2D system \eqref{Eq:FMII} is detectable and isolable according to  \cite{Malek_3DFDI} if $C\ssp{W}_{1}^*\cap C\ssp{W}_{2}^* =0$, where $\ssp{W}^*_1$ and $\ssp{W}^*_2$ denote the minimal \underline{finite} conditioned invariant subspace that contain $\ssp{L}_1^1+\ssp{L}_1^2$ and $\ssp{L}_2^1+\ssp{L}_2^2$, respectively (refer to \cite{Malek_3DFDI}, Theorem 2). It should be pointed out that the observability assumption of ($C$, $A_1$, $A_2$) is a fundamental requirement and condition  in \cite{Malek_3DFDI} (although it was stated in \cite{Malek_3DFDI} that this assumption was made for simplicity of the presentation). The main reason for the above limitation lies on and is due to the fact that the proposed approach in \cite{Malek_3DFDI} is based on the results of \cite{Massoumnia1986}. However, as stated in \cite{Massoumnia1986} the observability assumption is  quite a crucial and critical condition. For further illustration and  clarification for the above serious concerns consider the following 2D system,
\begin{equation}
	\begin{split}
		x(i+1,j+1) &= \bbm 0 &1\\0 &0\ebm x(i,j+1) + \bbm 0 &0\\0 &1\ebm x(i+1,j)+ \bbm 1\\0\ebm f_1+ \bbm 0\\1\ebm f_2,\\
		y(i,j) &= \bbm 0 &1\ebm x(i,j).
	\end{split}
\end{equation}
We have $\ssp{W}^*_1 = \ssp{L}_1=\ker C$, $\ssp{W}^*=\ssp{L}_2$, and consequently $C\ssp{W}_1^*\cap C\ssp{W}^*_2 = 0$. Therefore, the sufficient condition for solvability of the FDI problem under the zero initial condition in \cite{Malek_3DFDI} (Theorem 2) is satisfied. However, it is easy to verify that $\ssp{L}_1\subseteq\ssp{N}$, and consequently $f_1$ is \underline{not even detectable} (in other words, $f_1$ has no effect on the output signal). Our proposed methodology  does not suffer from the above limitation and restriction.

We are now in a position to state the following theorem.
\begin{theorem}\label{Thm:Our_Maleki}
	Consider the 2D system \eqref{Eq:FMII} and assume that the FDI problem is solvable by using the approach that is proposed in \cite{Malek_3DFDI}. Then the approach that we have proposed in this work \underline{can also} detect and isolate the faults in the system \eqref{Eq:FMII}.
\end{theorem}
\begin{proof}
	According to Theorem 3 in \cite{Malek_3DFDI}, $C\ssp{W}^*_1\cap C\ssp{W}^*_2=0$, and  $\ssp{W}^*_1$ and $\ssp{W}^*_2$ are internally/externally stabilizable.  Therefore, there exist two maps $D_1$ and $D_2$ such that $\ssp{W}_1^*$ and $\ssp{W}_2^*$ are both  $[A+DC]_{1,2}$-invariant, and the pair ($A_1+D_1C$, $A_2+D_2C$) is stable (that is, the corresponding 2D system is asymptotically stable). Since $C\ssp{W}^*_1\cap C\ssp{W}^*_2=0$, there exists a map $H$ such that $H_1C\ssp{W}^*_1 = C\ssp{W}^*_1$ and $H_1C\ssp{W}^*_2 = 0$. It follows that $\ssp{N}_s^h\cap\ssp{W}^*_1=0$, where $\ssp{N}_s^h$ is the invariant unobservable subspace of ($H_1C$, $A_1+D_1C$, $A_2+D_2C$). Note that $\ssp{N}_s^h$ is an unobservability subspace of ($C$, $A_1$, $A_2$) containing  $\ssp{L}_2^1+\ssp{L}_2^2$, and since $\op{S}_1^*$ is the smallest unobservability subspace containing $\ssp{L}_2^1+\ssp{L}_2^2$, it follows that $\op{S}_1^*\cap(\ssp{L}_1^1+\ssp{L}_1^2)=0$.
	Also, since  ($A_1+D_1C$, $A_2+D_2C$) is stable, it can be shown that ($A_1^p$, $A_2^p$) is also stable, where $P[A_1+D_1C]=A_i^pP$ and $P$ is the canonical projection of $\op{S}_1^*$. Therefore, one can construct an observer for the residual system \eqref{Eq:Filter_F1} to detect and isolate the fault $f_1$ (by choosing $D_{o1}=D_{o2}=0$). By following along the same procedure, one can detect and isolate the fault $f_2$. This completes the proof of the theorem.
\end{proof}

\begin{remark}\label{Rem:Our_Others}
	Theorem \ref{Thm:Our_Maleki} shows that our proposed approach can detect and isolate faults that are detectable and isolable by using the geometric method in \cite{Malek_3DFDI}.
	However, as shown below an example is provided where \underline{this approach fails} whereas our proposed approach \underline{can still} detect and isolate the faults. Also, we show that in this example the methods in \cite{BisiaccoMultiDim,BisiaccoLetter} are \underline{also not applicable}.
\end{remark}
\subsubsection{Illustrative Example (Limitations of the Methods in \cite{BisiaccoMultiDim,BisiaccoLetter} and \cite{Malek_3DFDI})}\ \\
Consider the 2D system \eqref{Eq:FMII} that is subjected to two faults $f_1$ and $f_2$. As stated earlier, in this work our main concern is on actuator faults and it is assumed that the output signals are not affected by faults (refer to Remark \ref{Rem:onFMII-General}).

Consider the 2D system \eqref{Eq:FMII} that is subjected to two faults $f_1$ and $f_2$ where,
\begin{equation}\label{Eq:ConterExm}
	\begin{split}
		A_1 &= \bbm \bsm 0 &0\\ 0 &0.5\esm &0.5I\\ 0 &0_{2\times2}\ebm,\; A_2 = 0.5\bbm 0_{2\times 2} &0\\ I &I\ebm,\;L_1^2=L_2^2=0,\; B_1=B_2=0\\
		L_1^1 &= [0,0,0,1]^\tran,\;L_2^1 = [0,0,-1,1]^\tran,\; C=\bbm 1 &0 &0 &0\\ 0 &0 &0 &1\ebm
	\end{split}
\end{equation}

The FDI problem is solvable by using the approach in \cite{BisiaccoMultiDim} if and only if,
\begin{equation}
	\rank(\bbm I-z_1A_1-z_2A_2 &z_1[L^1_1,L^1_2]\\
	C	&0\ebm) = n\;,\;\;\forall z_1,z_2\in\fld{C}
\end{equation}
The full rankness of the PBH matrix \eqref{Eq:PBHMatrix} is the necessary condition for the above property \cite{BisiaccoMultiDim} (page 231-item (i)).
However, for the 2D system \eqref{Eq:ConterExm}, one gets $\rank(PBH(1,0))=3<4$. Therefore, the FDI problem is \underline{not} solvable by \underline{using the method in \cite{BisiaccoMultiDim}}.

To check the conditions of Theorem \ref{Thm:BisiacoLetter}, let us first obtain $N(z_1,z_2)$. Consider $[a,b,c,d,e,f]$ as a row of $N(z_1,z_2)$. Since $N(z_1,z_2)PBH(z_1,z_2)=0$, it follows that,
\begin{equation}
	\left\{\begin{split}
		&a-0.5z_2c+e =0\\
		&-0.5z_1a+(1-0.5z_2)c = 0
	\end{split}\right.\;\;,\;\; \left\{\begin{split}
	&0.5z_1b+(1-0.5z_2)d + f =0\\
	&(1-0.5z_1)b-0.5z_2d = 0
\end{split}\right..
\end{equation}
From the first set of equations, one choice for $a$, $c$ and $e$ is given by $a= 2-z_2$, $c=z_1$ and $e = 0.5z_2z_1+z_2-2$, respectively. Also, based on the second set one can write $b=z_2$, $d=2-z_1$ and $f=2-z_2-z_1$. It can be shown that $N(z_1,z_2)=\bbm a &0 &c &0 &e &0\\ 0 &b &0 &d &0 &f\ebm$. Therefore, $N(z_1,z_2)\bbm z_1[L^1_1,L^1_2]\\
0\ebm= z_1\bbm 0 &z_1\\ 2-z_1 &2-z_1\ebm$. It is easy to check that the $\rank(N(z_1,2)\bbm z_1[L^1_1,L^1_2]\\
0\ebm)<2$ for all $z_2\in\fld{C}$ and $z_1 = 2$. Hence, the necessary condition that is \underline{proposed in \cite{BisiaccoLetter}} (as given by equation \eqref{Eq:BisiaccoLetterCon}) is \underline{not satisfied}.

Furthermore, the necessary condition to detect and isolate the fault $f_1$ by using the approach in \cite{Malek_3DFDI} is $C\ssp{W}_1^*\cap C\ssp{W}_2^*=0$. Since $\ssp{L}_1^1,\ssp{L}_2^1\not\in\ker C$, by invoking the algorithm \eqref{Eq:CIAlg}, one obtains $\ssp{W}_{1}^*=\ssp{L}_1$ and $\ssp{W}_{2}^*=\ssp{L}_2$. It follows that, $C\ssp{W}_{1}^*\cap C\ssp{W}_{2}^* =\spanset{[0,\;1]^\tran}$. Therefore, the \underline{necessary condition in \cite{Malek_3DFDI}} is also \underline{not} satisfied. In other words, the fault $f_1$ \underline{cannot} be detected and isolated by using the detection filter \eqref{Eq:Filter}, if one restricts the filter to the case with $M=C$ (or $H=I$),  according to the required results in\cite{Malek_3DFDI}.

Finally, it is now shown and demonstrated that one can detect and isolate \underline{both} faults $f_1$ and $f_2$ by using \underline{our proposed} methodology. Towards this end, by invoking the algorithm \eqref{Eq:Unob_MainAlg}, one can write $\op{S}_1^* = \ssp{L}_2^1$ (that is, the finite unobservability subspace containing $\ssp{L}_2^1$) that satisfies the condition \eqref{Eq:NecSufCon}. By considering $D_1=\bbm 0 &0 &0 &0\\ 0.5 &-0.5 &0 &0\ebm^\tran$ and $D_2=\bbm 0 &0 &0 &0\\ 0 &0 &0.5 &-0.5\ebm^\tran$, $\ssp{W}^*_2$ is $[A+DC]_{1,2}$-invariant. Also, since  $\ker C+\op{S}_1^* = \ker H_1C$ and $MP_2=HC$, one gets $H_1=[1,\;0]$ and $M_1=[1,0,0]$. Hence, the residual \eqref{Eq:Res} that is only affected by the fault $f_1$ is given by $A_1^p = \bs\bbm 0 &0 &\frac{\sqrt{2}}{2}\\ 0 &0.5 &0\\ 0 &0 &0\ebm\es$, $A_2^p = \bs\bbm  0 &0 &0\\  0 &0 &0\\ \frac{\sqrt{2}}{2} &\frac{\sqrt{2}}{2} &0.5 \ebm\es$ and $B_1=B_2=0$.
Since the pair ($A_1^p$, $A_2^p$) is stable, by considering $D_{o1}=D_{o2}=0$, the detection filter for the fault $f_1$ (as given by equation \eqref{Eq:Filter}) is now obtained according to,
\begin{equation}
	\begin{split}
		\omega(i+1,j+1)&=A_1^p\omega(i,j+1)+ A_2^p\omega(i+1,j),\\
		r_1(i,j)&=M_1\omega(i,j)-H_1y(i,j).
	\end{split}
\end{equation}
By following along the same procedure, one can also design a detection filter to detect and isolate the fault $f_2$. Therefore, \underline{our proposed approach can} accomplish the FDI objectives while the approaches that are proposed in \underline{\cite{BisiaccoMultiDim, BisiaccoLetter} and \cite{Malek_3DFDI} cannot} achieve this goal.
\begin{remark}\label{Rem:Generic}
	All the conditions for solvability of the FDI problem in the literature (for both 1D and 2D systems) and also our proposed conditions are \underline{{\bf generic}}, although this fact is not explicitly mentioned. For clarification, consider the faulty model \eqref{Eq:FMII}, where $x\in\fld{R}^2$, $A_1=A_2=0.4*I$, $B_1^1=L_1^1=[1,1]^\tran$, $B_2^2=L_2^2=[0,1]$ and $C = [1,-1]^\tran$. Let the initial condition ${\bf x}(0)=0$ and $f_1(i,j)=1$ for all $i+j\geq0$. It follows that $y\equiv0$, and consequently $f_1$ is not detectable. However, below we show that  sufficient conditions in the literature \cite{BisiaccoLetter,BisiaccoMultiDim,Malek_3DFDI}  as well as our proposed conditions are still all satisfied.
	\begin{enumerate}
		\item (Conditions in \cite{BisiaccoLetter} and \cite{BisiaccoMultiDim}). It follows that  $N(z_1,z_2)PHB(z_1,z_2)=0$, where $N(z_1,z_2)=\bbm -1 &0 &0.4(z_1+z_2)-1\\ 0 &1 &0.4(z_1+z_2)-1\ebm$, and consequently the condition in Theorem \ref{Thm:BisiacoLetter} is also satisfied.
		\item (Conditions in \cite{Malek_3DFDI}). By following the algorithm \eqref{Eq:CIAlg} we obtain $\ssp{W}_1^*=\ssp{L}_1$ and $\ssp{W}_2^*=\ssp{L}_2$. It follows that $C\ssp{W}_1^*\cap C\ssp{W}_2^*=0$, and consequently the condition in \cite{Malek_3DFDI} is also satisfied.
		\item (Our proposed Conditions)  By following the algorithm \eqref{Eq:Unob_MainAlg} we obtain $\op{S}_1^*=\ssp{L}_1$ and $\op{S}_2^*=\ssp{L}_2$. It follows that $\op{S}_1^*\cap \ssp{L}_2=0$, and consequently the condition \eqref{Eq:NecSufCon} is also satisfied.
	\end{enumerate}
\end{remark}
To summarize, in this section we have developed and presented a solution to the FDI problem of 2D systems by invoking an Inf-D framework for the first time in the literature and by utilizing invariant subspaces and derived necessary and sufficient conditions for solvability of the problem. It was shown that if the sufficient conditions for solvability of the FDI problem that are provided in \cite{Malek_3DFDI,Malek_3DFDIConf} are satisfied, then our proposed approach can also detect and isolate the faults. However, as shown above  there are certain systems that the methods in \cite{Malek_3DFDI,BisiaccoMultiDim,BisiaccoLetter}  are not applicable and capable of detecting and isolating faults, whereas our proposed approach can solve this problem successfully.
\section{Simulation Results}\label{Sec:Simulation}
In this section, we apply our proposed  FDI methodology to a two-line parallel heat exchanger system \cite{HyperPDE_2D,HEleakageThesis}. Specifically, we verify the necessary and sufficient conditions that are derived in the previous section, and also design a set of filters to detect and isolate  faults under both full and partial state measurement scenarios (this is to be realized by an appropriate selection of the output matrix $C$).

The heat exchanger is usually subject to two different types of faults, namely the fouling and the leakage\cite{HEleakageThesis}. 
The mathematical model of a typical heat exchanger is governed by the following hyperbolic PDEs,
\begin{equation}\label{Eq:HEfaulty}
\begin{split}
\frac{\partial T_f}{\partial t} &= -\alpha_f \frac{\partial T_f}{\partial z} - \beta(T_f - T_g) -  f_2(z,t),\\
\frac{\partial T_{g}}{\partial t} &= -\alpha_{g} \frac{\partial T_{g}}{\partial z} - \beta(T_{g} - T_f) +  f_1(z,t) +f_2(z,t),
\end{split}
\end{equation}
where $T_f$ and $T_g$ denote the temperature of the cold (fuel)  and the hot (exhaust gas) sections, respectively. The coefficients $\alpha_f$ and $\alpha_{g}$ are proportionally dependent on the speed of the fluid and the gas, respectively, and the coefficient $\beta$ is related to the heat transfer coefficient of the wall \cite{HEleakageThesis}. Moreover, $f_1$ and $f_2$ denote the leakage and the fouling effects, respectively. Finally, it is  assumed that the boundary conditions (the inlet temperature) $T_f(t,0)$ and $T_g(t,0)$ are given and only the outer section (i.e. $T_g$) is subject to the leakage.

\subsection{The Approximation of Hyperbolic PDE Systems by 2D Models}\label{Sec:PDE2FMII}
Let us first illustrate and demonstrate how one can approximate a general hyperbolic PDE system by using 2D models and representations. Consider the following hyperbolic PDE system
\begin{equation}\label{Eq:GeneralPDE}
\pardiff{\tilde{x}}{t} = \tilde{A}_1\pardiff{\tilde{x}}{z} + \tilde{A}_2\tilde{x} + \tilde{B} {u} + \sum_{k=1}^p \tilde{L}_k\tilde{f}_k,
\end{equation}
where  $z$ denotes the spatial coordinate, $\tilde{x}(z,t)\in\fld{R}^n$, ${u}(z,t)\in\fld{R}^q$  and $\tilde{f}_k(z,t)\in\fld{R}$ denote the state, input and fault signals, respectively. Also, the operators $\tilde{A}_1$, $\tilde{A}_2$, $\tilde{B}$ and $\tilde{L}_k$'s are real matrices with appropriate dimensions. Note that every hyperbolic PDE system with constant coefficients can be represented in the form \eqref{Eq:GeneralPDE} with a diagonalizable $\tilde{A}_1$ \cite{SiamPDEBook} (Chapter 1, Detention 1.1.1).

By applying the finite difference method to the system \eqref{Eq:GeneralPDE}, one obtains
\bs
\begin{equation}\label{Eq:GeneralPDEApr}
\begin{split}
\frac{\tilde{x}(i\Delta z, (j+1)\Delta t)-\tilde{x}(i\Delta z, j\Delta t)}{\Delta t} = &\tilde{A}_1 \frac{\tilde{x}(i\Delta z, j\Delta t)-\tilde{x}((i-1)\Delta z, j\Delta t)}{\Delta z} + \tilde{A}_2 \tilde{x}(i\Delta z, j\Delta t) \\&+ \tilde{B} \tilde{u}(i\Delta z, j\Delta t)+\sum_{k=1}^p \tilde{L}_k\tilde{f}_k(i\Delta z, j\Delta t).
\end{split}
\end{equation}
\es
By setting $x(i,j)= \bbm\tilde{x}((i-1)\Delta z,j\Delta t) \\ \tilde{x}(i\Delta z,j\Delta t)\ebm$, we can now write
\begin{equation}\label{Eq:GeneralFMIIApr}
\begin{split}
x(i+1,&j+1) =  A_1x(i,j+1) + A_2x(i+1,j) + B_1 u(i+1,j) +\sum_{k=1}^p {L}_{k}{f}_k(i+1,j),
\end{split}
\end{equation}
where $u(i+1,j)=\tilde{u}(i,j)$ for all $i$ and $j$, $A_1 = \bbm 0 &I\\ 0 &0\ebm$ and
\begin{equation*}
\begin{split}
A_2= \bbm 0 &0\\ -\frac{\Delta t}{\Delta z} \tilde{A}_1  &(I+\frac{\Delta t}{\Delta z} \tilde{A}_1 + \Delta t \tilde{A}_2)\ebm,\;
B_1 = \bbm 0\\ \tilde{B}\ebm, \; L_k=\bbm 0\\ \tilde{L}_k\ebm, \; f(i,j)=\tilde{f}(i+1,j).
\end{split}
\end{equation*}

Therefore, the PDE system \eqref{Eq:GeneralPDE} can be approximated by the FMII 2D model \eqref{Eq:GeneralFMIIApr}.
\begin{remark}
	As shown in \cite{Parabolic3D}, the \underline{parabolic} PDE systems can be approximated by FMII 3D models. As shown in Appendix \ref{Sec:3dFM}, our proposed FDI methodology for 2D systems can be extended to 3D systems, therefore the extension of the results of our proposed FDI methodology can also be applied to parabolic PDE systems.\qed
\end{remark}

For the purpose of conducting simulations, the parameters in equation \eqref{Eq:HEfaulty} are taken as $\alpha_f=\alpha_g=\beta=1$. Also, by considering $\Delta z=\Delta t=0.1$ one can discretize the system \eqref{Eq:HEfaulty} as,
\bs
\begin{equation}\label{Eq:FMIISimComplete}
\begin{split}
x(i+1,j+1) = &\bbm 0_{2\times2}  &I_{2\times2}\\ 0_{2\times2} &0_{2\times2}\ebm x(i,j+1) +
\bbm 0_{2\times2} &0_{2\times2}\\
I_{2\times2} &\bsm -0.1 &0.1\\0.1 &-0.1\esm
\ebm x(i+1,j)+ L_{1}^1 f_1(i+1,j)\\&+ B_2^1 u(i+1,j) + L_{2}^1 f_2(i+1,j),\\
y(i,j) = &Cx(i,j),
\end{split}
\end{equation}\es
where $\bs L_{1}^1 = [0,0,0,1]^\tran\es$, $\bs L_{2}^1=[0,0,-1,1]^\tran\es$ and $\bs x(i,j) = [T_g((i-1)\Delta z,j \Delta t), T_f((i-1)\Delta z,j \Delta t),\es$
$ \bs T_g(i\Delta z,j \Delta t), T_f(i\Delta z,j \Delta t)]^\tran\es$. Finally, by following along the same steps as in \cite{ACC2013}, one obtains $B_2 = \bbm 0_{2\times2} \\  \bsm 1.1 &-0.1 \\ -0.1 &1.1\esm\ebm$, $\bs u(0,j) =[T_f(0, \Delta t),T_g(0, \Delta t)]^\tran\es$ and $u(i,j)=u(0,j)$ for all $i$.

We first assume that both temperatures $T_f$ and $T_g$ are available for measurement along the spatial coordinates at discrete points (i.e., $T_f(i\Delta z,j\Delta t)$ and $T_g(i\Delta z,j\Delta t)$ are available from sensors). Next, we consider the case where only the outer temperature (that is, $T_g$) is available for measurement. As we shall show subsequently, in the latter case by using the 1D approximate ODE model (for example, as in \cite{BoundaryControl_1DApp}), the faults $f_1$ and $f_2$ are \underline{not} isolable, whereas by using our proposed 2D FDI  methodology one can detect and isolate  both faults.
\subsection{FDI of a Heat Exchanger by Using Full State Measurements}
Let us assume that both temperatures (namely, $T_f$ and $T_g$ ) are available for measurement. Therefore, one can select the output matrix as $C = \bbm 1 &0 &0 &0\\ 0 &0 &0 &1\ebm$.

\subsubsection{Detectability and Isolability Conditions}\ \\
By applying the algorithm \eqref{Eq:CIAlg}, the minimal finite conditioned invariant subspace containing $\ssp{L_2^1}=\spanset{L_2^1}$ is obtained as $\ssp{W}_{2}^* = \ssp{L}_2^1$, and by applying the algorithm \eqref{Eq:Unob_MainAlg}, one obtains $\op{S}^*=\ssp{W}_{2}^*=\ssp{L}_2^1$. It follows that $\ssp{L}_1^1\cap\op{S}^* =0$. Therefore, the sufficient condition (as given by Theorem \ref{Thm:NecSuffFDI}) is satisfied. By following along the same lines the necessary and sufficient condition for detectability and isolability of the fault $f_2$ can also be shown to be satisfied.

\subsubsection{FDI 2D Luenberger Filter Design}\ \\
As stated earlier, we are interested in designing 2D Luenberger detection filters by using the LMI condition that is proposed in the Subsection \ref{Sec:Observer}. In this part of the paper, we design a filter for detecting and isolating the fault $f_1$ (without loss of any generality, by following along the same lines as conducted below one can also design a filter to detect and isolate the fault $f_2$). The 2D detection filter must be decoupled from the fault $f_2$ (refer to the conditions in equation \eqref{Eq:FDIP}). As stated above, the finite unobservability subspaces containing the subspace $\ssp{L}_2^1$ is obtained by using the algorithm \eqref{Eq:Unob_MainAlg} and is given by $\op{S}^{*} = \mathscr{W}_{2}^{*}=\ssp{L}_2^1$.

The output injection matrices $D_1$ and $D_2$ for $\mathscr{W}_{2}^{*}$ are to be derived such that $\mathscr{W}_{2}^{*}$ is $[A+DC]_{1,2}$-invariant. Therefore, one can write
$(A_1+D_1C) W_{2}^{c}= 0$ and $(A_2+D_2C) W_{2}^{c}= 0$, where the columns of $W_{2}^{c}$ are the basis of $\ssp{W}_{2}^*\cap(\ssp{W}_{2}^*\cap\ker C)^\bot$. One solution to $D_1$  and $D_2$ is $D_1=\bbm 0 &0 &0 &0\\1 &-1 &0&0\ebm^\tran$ and $D_2 = \bbm 0 &0 &0 &0\\0 &0 &-0.2&0.2\ebm^\tran$. 
Also, let $P=\bbm 1 &0 &0 &0\\ 0 &1 &0 &0\\ 0 &0 &1 &1\ebm$ (which is the canonical projector of the subspace $\op{S}_1^{*}$), where $P$ is used in equation \eqref{Eq:Filter_F1}. By using $\ker H_1C=\op{S}_1^*+\ker C$ and $M_1P_1=H_1C$, we have $H=[1,\;0,\;0]$ and the output matrix $M$ becomes $M=[1,0,0]$.  Hence, the factored out 2D system is now expressed as,
\bs
\begin{equation*}\label{Eq:SimQuSubSys}
\begin{split}
\omega(i+1,j+1) &= \left[ \bsm 0 &0 &-1.42\\ 0 &0 &0\\0 &0 &0\esm\right] \omega(i,j+1) + \left[ \bsm 0 &0 &0\\0 &0 &0\\ -0.7 &-0.7 &0\esm\right] \omega(i+1,j) + \bbm \bsm 0 &0\\0 &0\\0.7 &0.7\esm \ebm u(i+1,j) + \bbm\bsm 0\\1\\0\esm\ebm f_1(i+1,j),\\
y_p(i,j) &= [1,0,0]\omega(i,j),
\end{split}
\end{equation*}
\es
where $\omega(i,j) = P_1 x(i,j)$. It is straightforward to show that the positive definite matrices  $R_1 = \diag(0.4, 1, 2.133)$ and $R_2 = \diag(0.4, 2.15, 0.86)$ satisfy the inequality  $W_{c_d}^\tran A_C W_{c_d}<0$. By using Remark \ref{Rem:ObserverLMIConstructive}, one can obtain $D_{o1}=0$ and $D_{o2} = \bbm 0 &0 &-0.7\ebm^\tran$. Therefore, the filter to detect and isolate the fault $f_1$ is given by,
\begin{equation*}\label{Eq:SimQuSubSys2}
\begin{split}
\omega(i+1,j+1) = &\left[ \bsm 0 &0 &-1.42\\ 0 &0 &0\\0 &0 &0\esm\right] \omega(i,j+1) + \left[ \bsm 0 &0 &0\\0 &0 &0\\ -0.7 &-0.7 &0\esm\right] \omega(i+1,j) + \bbm \bsm 0 &0\\0 &0\\0.7 &0.7\esm \ebm u(i+1,j) \\ &+ D_{o2}P_2y(i,j+1),\\
r_1(i,j) &= [1,0,0]\omega(i,j)-Hy(i,j).
\end{split}
\end{equation*}
Designing a filter to detect and isolate the fault $f_2$ follows  along  the same lines as those given above for the fault $f_1$. These details are not included for brevity. 
\subsubsection{Threshold Computation}\ \\
Due to presence of input and output noise, disturbances, and uncertainties in the model, the value of the residual $r_k(i,j)$ is not exactly equal to zero under the fault free situation.
Therefore, to reduce the number of false alarm flags one needs to apply threshold bands to the residual signals. In this subsection, we present an approach for determining  the thresholds that are needed for achieving the FDI task. For 1D systems, there are a number of approaches for computing a threshold, e.g. based on the maximum or the root mean square (RMS) of the residual signal \cite{ThresholdCal}. In this work, we use the maximum residual norm. However, one can also apply the RMS approach to 2D systems.

Consider the FMII 2D model \eqref{Eq:FMII} subject to the fault free situation. The threshold $th_k$ is then determined from,
\begin{equation}
th_k = \begin{aligned}
& \underset{\ell<N_0}{\text{Max}}
& & |r_k^\ell(i,j)|
\end{aligned}, \;\; \mathrm{for}\; i,j\leq N_1
\end{equation}
where $|\cdot|$ denotes a norm function (in this work we use the norm-2), $N_0$ is the number of the Monte Carlo simulations (refer to \cite{MonteCarloBook}) that are used to determine the threshold, $r_k^\ell(\cdot,\cdot)$ denotes the signal of $k^{\mathrm{th}}$ residual in the $\ell^{\mathrm{th}}$ length and $N_1$ is a sufficiently large number.

By utilizing  predefined thresholds that are denoted by $th_k, k=1,2$, the FDI logic can be summarized as follows,
\begin{equation}\label{Eq:FDILogic}
\begin{split}
\mathrm{if}\;r_1 >th_1 \; \Rightarrow \; \mathrm{the\;fault}\; f_1\;\mathrm{has\;occurred}.\\
\mathrm{if}\;r_2 >th_2 \; \Rightarrow \; \mathrm{the\;fault}\; f_2\;\mathrm{has\;occurred}.\\
\end{split}
\end{equation}

Let us now consider two scenarios. In the first scenario, a single fault $f_1$ with the severity of $1$ occurs at ($i=5$ and $j\geq60$). In the second scenario, multiple faults $f_1$ and $f_2$ with the severity of $1$  occur at ($i=5$ and $j\geq50$) and ($i=5$ and $j\geq70$), respectively. The residuals $r_1$ and $r_2$ for the first scenario  are shown in Figure \ref{fig:ScenarioI}. For the second scenario, the results are shown in Figure \ref{fig:ScenarioII}. The thresholds are determined by conducting  Monte Carlo simulations \cite{MonteCarloBook} corresponding to the healthy 2D system. As shown in Figures \ref{fig:ScenarioI} and \ref{fig:ScenarioII}, our proposed methodology can detect and isolate the faults in both single- and multiple-fault scenarios (according to the FDI logic that is given by equation \eqref{Eq:FDILogic}).
\begin{figure}
	\centering
	\begin{subfigure}{0.4\textwidth}
		\includegraphics[scale = 0.45]{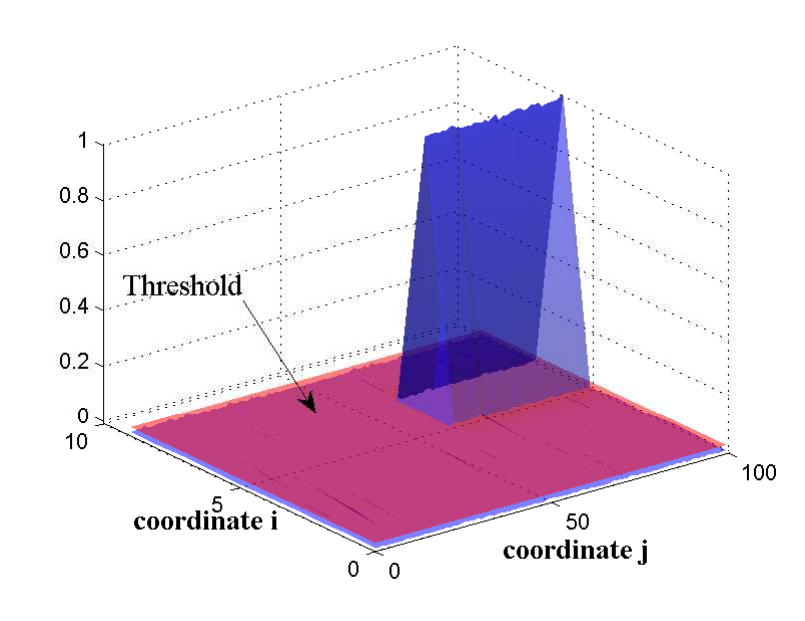}
		\caption{The residual signal $r_1$ for detecting and isolating the fault $f_1$.}
		\label{fig:FaultFl_Ff}
	\end{subfigure}
	~
	\begin{subfigure}{0.4\textwidth}
		\includegraphics[scale = 0.45]{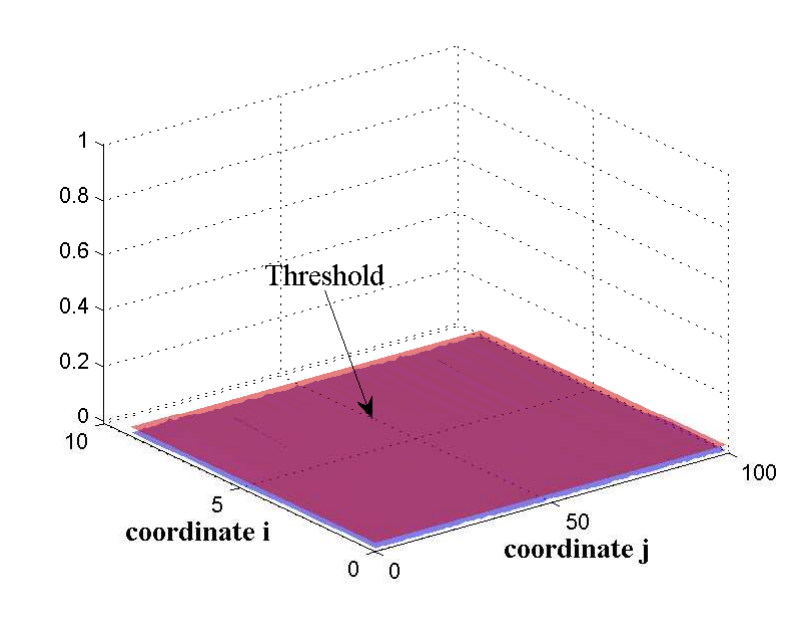}
		\caption{The residual signal $r_2$ for detecting and isolating the fault $f_2$.}
		\label{fig:FaultFl_Fl}
	\end{subfigure}
	\caption{The residual signals for detecting and isolating the fault $f_1$ for the first scenario (a single fault case) using full state measurement.}
	\label{fig:ScenarioI}
	\vspace{-2mm}
\end{figure}

\begin{figure}
	\centering
	\begin{subfigure}{0.45\textwidth}
		\includegraphics[scale=0.45]{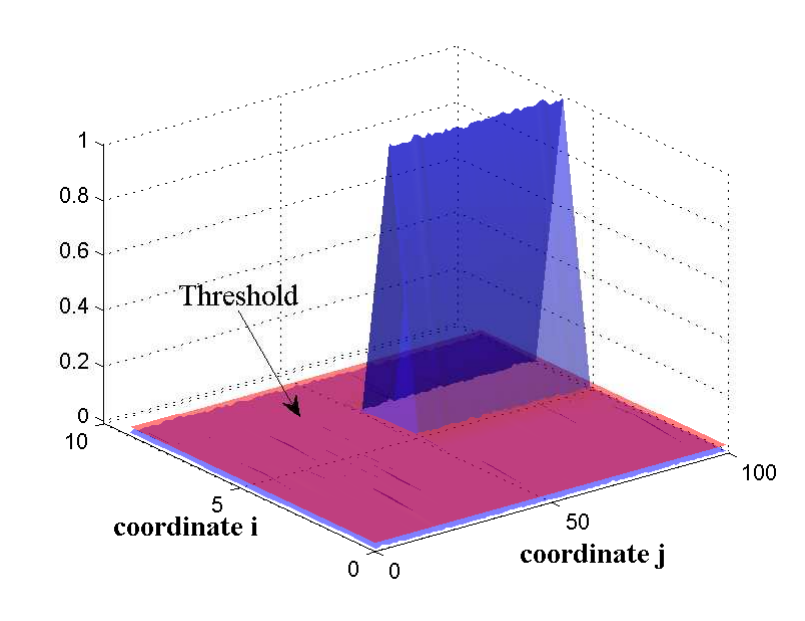}
		\caption{The residual signal $r_1$ for detecting and isolating the fault $f_1$.}
		\label{fig:FaultBoth_Ff}
	\end{subfigure}
	\begin{subfigure}{0.45\textwidth}
		\includegraphics[scale=0.45]{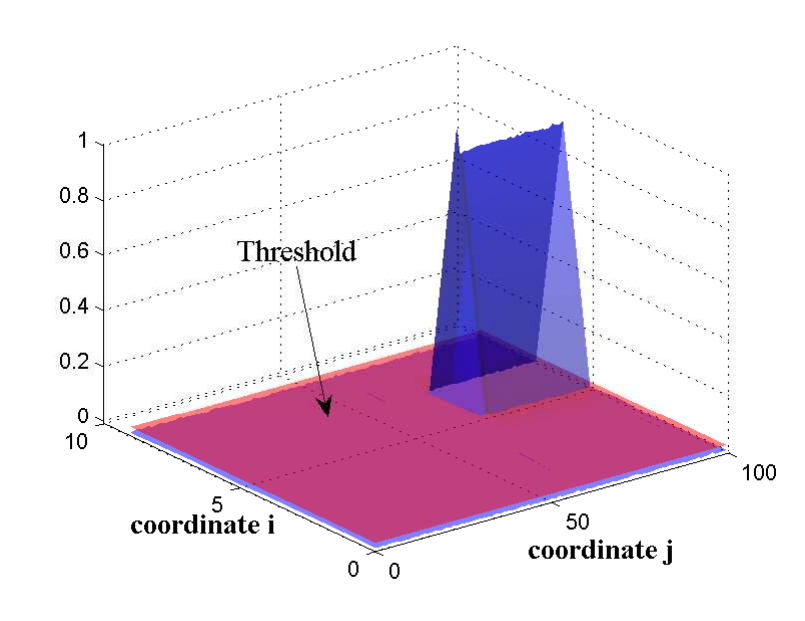}
		\caption{The residual signal $r_2$ for detecting and isolating the fault $f_2$.}
		\label{fig:FaultBoth_Fl}
	\end{subfigure}
	\caption{The residual signals for detecting and isolating the faults $f_1$ and $f_2$ for the second scenario (a multiple fault case) using full state measurement.}
	\label{fig:ScenarioII}
	\vspace{-2mm}
\end{figure}
\subsection{FDI of a Heat Exchanger Using Partial State Measurements}
In this section, we assume that only the outer temperature ($T_g$) is available for measurement. This corresponds to a more practical and physically feasible scenario in various applications (sensing the inner temperature requires a more sophisticated hardware). In this case, we set $C = \bbm 0 &1 &0 &0\\0 &0 &0 &1\ebm$ (in other words, we measure $T_g(i-1,j)$ and $T_g(i,j)$). In this subsection, we demonstrate and illustrate the capabilities of our proposed FDI approach based on the 2D system modeling, whereas it is shown that by using a \underline{1D approximation} of the PDE system \eqref{Eq:HEfaulty} the FDI problem \underline{cannot} be solved.
\subsubsection{1D Approximation of the Heat Exchanger System}\ \\
The hyperbolic PDE system \eqref{Eq:HEfaulty} can be approximated by applying discretization through the $z$ coordinates as follows. Let $\ell$ denote the length of the heat exchanger that is discretized  into $N$ equal intervals (i.e., $\Delta z = \frac{\ell}{N}$). By using the approximation $\pardiff{T_g(k\Delta z,t)}{z} = \frac{T_g(k\Delta z,t)-T_g((k-1)\Delta z, t)}{\Delta z}$, one can represent the PDE system \eqref{Eq:HEfaulty} by the following 1D approximate model,
\begin{equation}\label{Eq:1DApproximation}
\begin{split}
\dot{x} &= Ax+Bu+\sum_{k=1}^N (L_1^k f_1^k+L_2^kf_2^k),\\
y &= Cx,
\end{split}
\end{equation}
where $x(t) = \bbm T_g(\Delta z,t), T_f(\Delta z,t),\cdots,T_g(N\Delta z, t),T_f(N\Delta z, t)\ebm^\tran\in\fld{R}^{2N}$, $B = \bbm 1,0,0,0,\cdots,0\\0,1,0,0,\cdots,0\ebm^\tran$, and $f_i^k(t) = f_1(k\Delta z,t)$ ($i=1,2$). Also, $L_1^k = [\underbrace{0,\cdots,0}_{(k-1)},-1,1,0,\cdots]^\tran$, and the fault signatures $L_2^k$  are  $2N$-dimensional vectors such that only the $k^\mathrm{th}$ element is 1 and the rest are zeros. Moreover,
\bs
\begin{equation}\label{Eq:1DApproximation_par}
\begin{split}
A = \bbm A_1 &0 &0  &\cdots\\
A_2 &A_1 &0  &\cdots\\
0 &A_2 &A_1  &\cdots\\
0 &0 &\ddots &\ddots &\ddots\ebm, C = \bbm 1 &0 &0 & 0&\cdots\\
0 &0 &1 &0 &\cdots\\
0  &0 &0 &0 &\ddots\ebm
\end{split}
\end{equation}
\es
in which $\bs A_1 = \bbm -\frac{1+\Delta z}{\Delta z} &1\\1 &-\frac{1+\Delta z}{\Delta z}\ebm\es$ and $\bs A_2 = \bbm -\frac{1}{\Delta z} &0\\0 &-\frac{1}{\Delta z}\ebm\es$. Now, consider the faults $f_1^k$ and $f_2^k$. It can be shown that $\ssp{W}^*_{2}=\spanset{L_1^k,L_2^k}$, where $\ssp{W}^*_{2}$ is the minimal conditioned invariant subspace containing $\ssp{L}_2^k$ (from the 1D system perspective). Consequently, $\op{S}^*_{1D} \cap \ssp{L}_1^k\neq 0$ ($\op{S}^*_{1D}$ denotes the unobservability subspace containing $\ssp{L}_2^k$ in the 1D system sense). Consequently, the faults $f_1^k$ and $f_2^k$ {\bf are not} isolable. However,  we show below that the faults can be detected and isolated if one approximates the system \eqref{Eq:HEfaulty} by using the 2D model representation.
\subsubsection{2D Representation of the Heat Exchanger}\ \\
Let us set $x_1(i,j) = T_g(i,j)+\frac{\Delta t}{\Delta z}f_1(i,j)$ and $x_2(i,j) = T_f(i,j)$, so that the system \eqref{Eq:HEfaulty} is approximated by the system \eqref{Eq:FMIISimComplete} where all the operators are defined as before except for $\bs L_1^1 = \bbm 0 &1 &0 &0\ebm\es$ and $\bs C= \bbm 0 &1 &0 &0\\0 &0 &0 &1\ebm\es$. Note that since only the state $T_g$ is assumed to be available for measurement, we sense $T_g(i-1,j)$ and $T_g(i,j)$. By applying the algorithm \eqref{Eq:Unob_MainAlg}, where $\ssp{L}=\ssp{L}_2^1$, one gets $\op{S}^{*} = \spanset{L_2^1,\bbm 1 &0 &0 &0\ebm^\tran,\bbm 0 &0 &1 &1\ebm^\tran}$. Since, $\ssp{L}_1^1\cap\op{S}^* = 0$, the fault $f_1$ is both detectable and isolable (according to  Theorem \ref{Thm:NecSuffFDI}). It should be noted that  by applying the approaches in \cite{BisiaccoMultiDim,BisiaccoLetter,Malek_3DFDI}, one can also detect and isolate the fault $f_1$. By following along the same lines as the ones given earlier one can show that the fault $f_2$ is also detectable and isolable, where one can  design the required detection filters. Figures \ref{fig2:ScenarioI} and \ref{fig2:ScenarioII} depict the simulation results for the scenarios that presented above. As can be observed from Figures \ref{fig:ScenarioII} and \ref{fig2:ScenarioII} (particularly, Figures \ref{fig:FaultBoth_Fl} and \ref{fig:FaultBoth_Fl_2}), it follows that fewer available information (recall that both $T_f$ and $T_g$ are measurable in Figure \ref{fig:ScenarioII}, whereas in Figure \ref{fig2:ScenarioII}, only $T_g$ is measurable) results in a situation where the spatial coordinate of a fault cannot be estimated accurately.  
\begin{figure}
	\centering
	\begin{subfigure}{0.4\textwidth}
		\includegraphics[scale = 0.45]{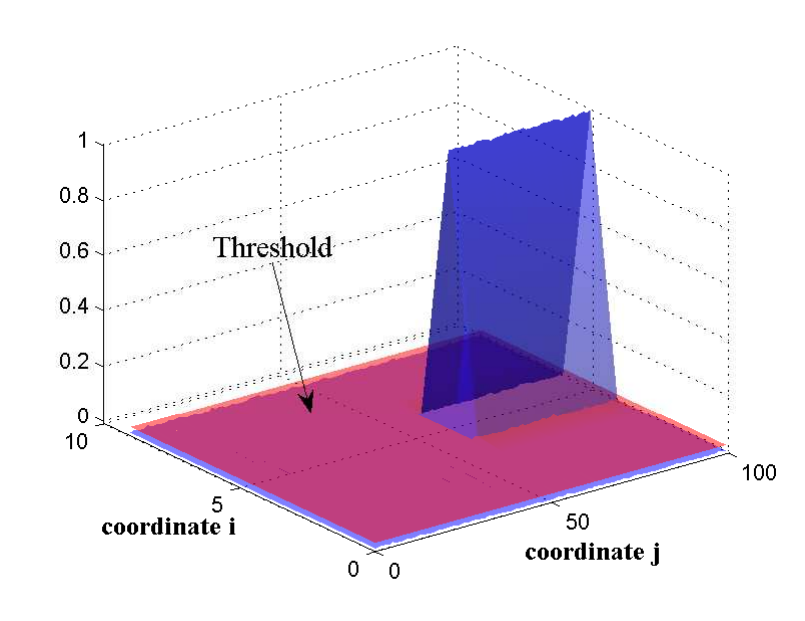}
		\caption{The residual signal $r_1$ for detecting and isolating the fault $f_1$.}
		\label{fig:FaultFl_Ff_2}
	\end{subfigure}
	~
	\begin{subfigure}{0.4\textwidth}
		\includegraphics[scale = 0.45]{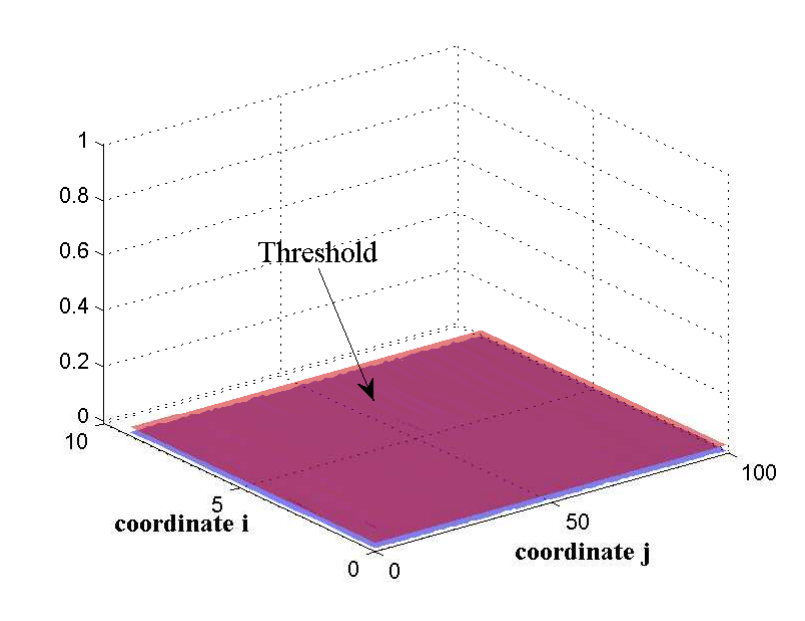}
		\caption{The residual signal $r_2$ for detecting and isolating the fault $f_2$.}
		\label{fig:FaultFl_Fl_2}
	\end{subfigure}
	\caption{The residual signals for detecting and isolating the fault $f_1$ for the first scenario (a single fault case) using partial state measurement.}
	\label{fig2:ScenarioI}
	\vspace{-2mm}
\end{figure}

\begin{figure}
	\centering
	\begin{subfigure}{0.45\textwidth}
		\includegraphics[scale=0.45]{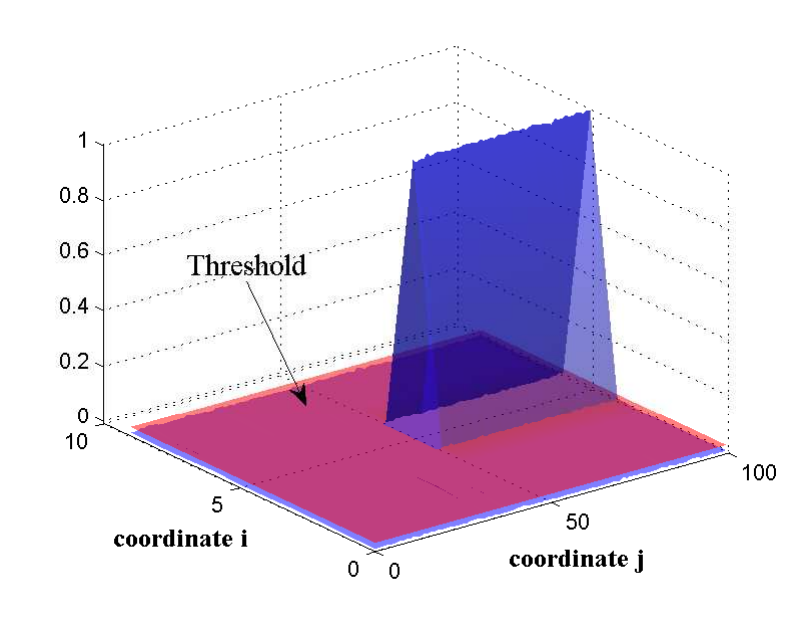}
		\caption{The residual signal $r_1$ for detecting and isolating the fault $f_1$.}
		\label{fig:FaultBoth_Ff_2}
	\end{subfigure}
	\begin{subfigure}{0.45\textwidth}
		\includegraphics[scale=0.45]{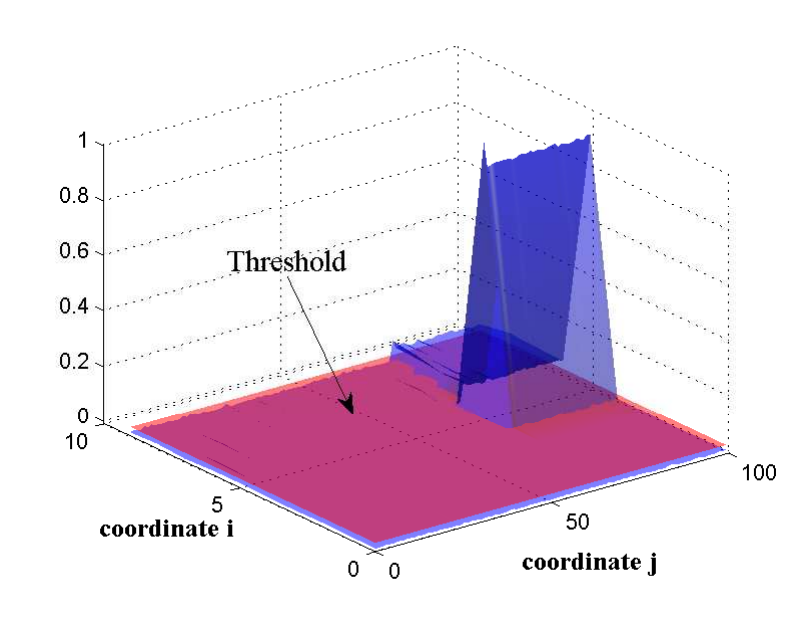}
		\caption{The residual signal $r_2$ for detecting and isolating the fault $f_2$.}
		\label{fig:FaultBoth_Fl_2}
	\end{subfigure}
	\caption{The residual signals for detecting and isolating the faults $f_1$ and $f_2$ for the second scenario (a multiple faults) using partial state measurement.}
	\label{fig2:ScenarioII}
	\vspace{-2mm}
\end{figure}
\vspace{-2mm}
\section{Conclusions}\label{Sec:Conc}
The fault detection and isolation (FDI) problem for discrete-time two-dimensional (2D) systems represented by the Fornasini-Marchesini model II (where other 2D models are subclasses  of this representation) is investigated in this work. In order to derive the necessary and sufficient conditions for solvability of the FDI problem, the notion of the conditioned invariant and unobservability subspace of 1D systems was generalized to 2D systems by using an infinite dimensional (Inf-D) framework and representation. Moreover, algorithms for computing and constructing these subspaces are introduced and provided that converge in a \underline{finite and known} number of steps. By applying the linear matrix inequality (LMI) approach, sufficient conditions for existence of an asymptotically convergent 2D state estimation observer is derived. Necessary and Sufficient conditions for  solvability of the FDI problem are also provided.  Furthermore, another set of sufficient conditions (that are based on  deadbeat observers) are also provided. It was shown that although the sufficient conditions for applicability of  the currently available geometric results in the literature are also sufficient for our proposed approach to accomplish the FDI goal, however, there are 2D systems where the geometric approaches in the literature \underline{are not applicable} to detect and isolate the faults, whereas our approach \underline{can still} achieve the FDI objective and task. Finally, simulation results are provided for the application of our proposed FDI methodology to an heat exchanger system to demonstrate and illustrate the capabilities and advantages of our proposed solution as compared to the alternative 1D representation and 1D FDI approaches.
\appendix
\section{Generalization to 3D Systems}\label{Sec:3dFM}
The results that are derived in this paper are dependent on invariant concepts that are
provided in previous sections. Since these invariant subspaces are independent of the order
of A1 and A2, there is a natural extension that is possible to FMII 3D models. In this
appendix, we briefly show how our previous results can be extended to 3D systems.
Furthermore, it should be pointed out that one can also extend these results to nD systems

Consider the following 3D FMII model
\bs
\begin{equation}\label{Eq:FM3D}
\begin{split}
&x(i+1,j+1,k+1)= A_1x(i,j+1,k+1)+ A_2x(i+1,j,k+1)\\&+A_3x(i+1,j+1,k)+ B_1u(i,j+1,k+1)\\&+ B_2u(i+1,j,k+1)+B_3u(i+1,j+1,k),\\
&y(i,j,k) = Cx(i,j,k).
\end{split}
\end{equation}
\es
The corresponding Inf-D system is given by
\begin{equation}\label{Eq:IDRep3D}
\begin{split}
{\bf x}(k+1) &= \op{A}{\bf x}(k),\;\;\;\; k\in\underline{\fld{N}} \\
{\bf y}(k) &=\op{C}{\bf x}(k),
\end{split}
\end{equation} 
where ${\bf x}(k)\in\op{X}=\sumbanach{\fld{R}^n}$, ${\bf y}(k)=(\cdots,y(-1+k,1)^\tran,y(k,0)^\tran,y(1+k,-1)^\tran,\cdots)^\tran\in\sumbanach{\fld{R}^q}$, and $\op{A}$ is an Inf-D matrix with $A_1$, $A_2$ and $A_3$ as diagonal and upper  diagonal blocks, respectively, with the remaining elements set to zero, and $\op{C} = \diag(\cdots,C,C,\cdots)$. In other words, we have,
\begin{equation}
\op{A} = \bbm \ &\ddots &\ddots &\cdots & & &\cdots\\
\cdots &0 &A_1 &A_2 &A_3 &0 &\cdots\\
\cdots &0 &0 &A_1 &A_2 &A_3 &\cdots\\
\cdots & &\cdots & & &\ddots &\ddots \ebm , 
\op{C} = \bbm \ddots & &\cdots & &\cdots\\
\cdots &0 &C &0 &\cdots\\
\cdots &0 &0 &C &\cdots \\
\cdots & &\cdots & &\ddots\ebm												
\end{equation}

The generalization of the Definition \ref{Def:A12_Inv} can be stated as follows.
\begin{definition}
	The subspace $\ssp{V}$ is called $A_{1,2,3}$-invariant if $A_i\ssp{V}\subseteq\ssp{V}$ for all $i=1,2,3$.\qed
\end{definition}
Also, it can be shown that the finite unobservable subspace of the 3D system \eqref{Eq:FM3D} is given by   $\ssp{N}=\bigcap_{i+j+k=0}^n \ker CA^{(i,j,k)}$,

where
\begin{equation*}
\begin{split}
A^{(i,j,k)} &= A_1 A^{(i-1,j,k)}+A_2 A^{(i,j-1,k)}+ A_3A^{(i,j,k-1)},\\
A^{(0,0,0)} &= I\;\;\; A^{(i,j,k)}=0\;\mathrm{if} \; i \; \mathrm{or} \; j\; \mathrm{or}\; k < 0.
\end{split}
\end{equation*}
In view of the above, one can now extend the invariant unobservable subspace $\ssp{N}_s$ to
$\ssp{N}_s = \bigcap_{||\alpha||>0}\ker CA^\alpha$,
where, $\alpha$ is a multi-index parameter, and $|\alpha|$ and $||\alpha||$ are defined according to the same manner that were defined for 2D systems (refer to the notation in the Introduction section). For instance, if $\alpha = (3,1,3,2,1,1,1,3,3,3)$, then $|\alpha| = (4,1,5)$ and $||\alpha|| = 10$. Also, we have $\ssp{N}_{s,\infty}=\sum\ssp{N}_s$.

By following along the same lines as in Theorem \ref{Thm:AID_A12Inv} and Lemma \ref{Lem:CA_Inv}, we have the following corollaries.
\begin{corollary}\label{Col:AID_A123}
	Consider the 3D system \eqref{Eq:FM3D} and the Inf-D system \eqref{Eq:IDRep3D}. Let $\ssp{V}_\infty = \sum \ssp{V}$, where $\ssp{V}\subseteq\fld{R}^n$. The subspace $\ssp{V}_\infty$ is $\op{A}$-invariant if and only if $\ssp{V}$ is $A_{1,2,3}$-invariant.
\end{corollary}
\begin{corollary}
	Consider the 3D system \eqref{Eq:FM3D} and the Inf-D system \eqref{Eq:IDRep3D}. The following statements are equivalents.
	\begin{enumerate}
		\renewcommand{\labelenumi}{(\roman{enumi})}
		\item The subspace $\ssp{W}_\infty$ is conditioned invariant.
		\item  $A_1(\ssp{W}\cap\ker C)+A_2(\ssp{W}\cap\ker C)+A_3(\ssp{W}\cap\ker C)\subseteq\ssp{W}$.
		\item $\op{A}(\ssp{W}_\infty\cap\ker\op{C})\subseteq\ssp{W}_\infty$.
	\end{enumerate} 
	where $\ssp{W}_\infty = \sum \ssp{W}$.
\end{corollary}

By following along the same steps as above, the invariant subspace and the results provided in Sections \ref{Sec:InvSpace} and \ref{Sec:FDI} can be extended to  nD systems. Therefore, one can generalize the results of this paper to nD systems without major barriers and challenges.

\bibliographystyle{IEEEtr}
\bibliography{2DJournalRef}

\begin{thebibliography}{10}

\bibitem{IsermannBook2006}
R.~Isermann, {\em Fault-diagnosis systems: an introduction from fault detection
  to fault tolerance}.
\newblock Springer, 2006.

\bibitem{Christo_Book}
P.~D. Christofides, {\em Nonlinear and robust control of PDE systems: Methods
  and applications to transport-reaction processes}.
\newblock Springer, 2001.

\bibitem{ACC2012}
A.~Baniamerian and K.~Khorasani, ``Fault detection and isolation of dissipative
  parabolic {PDE}s: Finite-dimensional geometric approach,'' in {\em 2012
  American Control Conference}, pp.~5894--5899, 2012.

\bibitem{Davis_Journal}
N.~H. El-Farra and S.~Ghantasala, ``Actuator fault isolation and
  reconfiguration in transport-reaction processes,'' {\em AIChE Journal},
  vol.~53, no.~6, pp.~1518--1537, 2007.

\bibitem{Christofer_Hyperbolic}
P.~Christofides and P.~Daoutidis, ``Robust control of hyperbolic {PDE}
  systems,'' {\em Chemical engineering science}, vol.~53, no.~1, pp.~85--105,
  1998.

\bibitem{Curtain_Book}
R.~F. Curtain and H.~Zwart, {\em An introduction to infinite-dimensional linear
  systems theory}, vol.~21.
\newblock Texts in Applied Mathematics, Springer-Verlag, Springer, New York,
  1995.

\bibitem{Krstic_Hyperbolic}
M.~Krstic and A.~Smyshlyaev, ``Backstepping boundary control for first-order
  hyperbolic {PDE}s and application to systems with actuator and sensor
  delays,'' {\em Systems \& Control Letters}, vol.~57, no.~9, pp.~750--758,
  2008.

\bibitem{HyperPDE_2D}
W.~Marszalek, ``Two-dimensional state-space discrete models for hyperbolic
  partial differential equations,'' {\em Applied Mathematical Modeling},
  vol.~8, pp.~11--14, 1984.

\bibitem{ACC2013}
A.~Baniamerian, N.~Meskin, and K.~Khorasani, ``Geometric fault detection and
  isolation of two-dimenional systems,'' in {\em 2013 American Control
  Conference}, pp.~3547--3554, 2013.

\bibitem{Parabolic3D}
P.~Stavroulakis and S.~Tzafestas, ``State reconstruction in low-sensitivity
  design of 3-dimensional systems,'' {\em IEE Proceedings of Control Theory and
  Applications, Part D}, vol.~130, no.~6, pp.~333--340, 1983.

\bibitem{Kaczorek_Book}
T.~Kaczorek, {\em Two-dimensional linear systems}.
\newblock Springer-Verlag, 1985.

\bibitem{2DRef}
G.~W. Pulford, ``The two-dimensional power spectral density: A connection
  between 2-{D} rational functions and linear systems,'' {\em IEEE Transactions
  on Automatic Control}, vol.~56, no.~7, pp.~1729--1734, 2011.

\bibitem{Stability2D2013}
S.~Knorn and R.~H. Middleton, ``Stability of two-dimensional linear systems
  with singularities on the stability boundary using lmis,'' {\em IEEE
  Transaction on Automatic Control}, vol.~58, no.~10, pp.~2579--2590, 2013.

\bibitem{FMinBook}
E.~Fornasini and G.~Marchesini, ``Properties of pairs of matrices and
  state-models for two-dimensional systems. part1: State dynamics and geometry
  of the pairs,'' {\em Multivariate analysis: Future directions}, vol.~5,
  pp.~131--180, 1993.

\bibitem{BisiaccoMultiDim}
M.~Bisiacco and M.~E. Valcher, ``Observer-based fault detection and isolation
  for 2{D} state-space models,'' {\em Multidimensional Systems and Signal
  Processing}, vol.~17, pp.~219--242, 2006.

\bibitem{FMResidual}
E.~Fornasini and G.~Marchesini, ``Residual generators for detecting failure in
  2{D} systems,'' in {\em Integrating Research, Industry and Education in
  Energy and Communication Engineering, MELECON'89}, pp.~69--72, 1989.

\bibitem{BisiaccoLetter}
M.~Bisiacco and M.~E. Valcher, ``The general fault detection and isolation
  problem for 2{D} state-space models,'' {\em Systems \& control letters},
  vol.~55, no.~11, pp.~894--899, 2006.

\bibitem{RepitativeProcess_Book}
E.~Rogers, K.~Galkowski, and D.~H. Owens, {\em Control systems theory and
  applications for linear repetitive processes}, vol.~349.
\newblock Lecture Notes in Control and Information Sciences, Springer, 2007.

\bibitem{ILC_2D_Tac}
D.~Meng, Y.~Jia, J.~Du, and S.~Yuan, ``Robust discrete-time iterative learning
  control for nonlinear systems with varying initial state shifts,'' {\em IEEE
  Transactions on Automatic Control}, vol.~54, no.~11, pp.~2626--2631, 2009.

\bibitem{Mass_Thesis}
M.~A. Massoumnia, {\em A geometric approach to failure detection and
  identification in linear systems}.
\newblock PhD thesis, MIT, Dep. Aero. \& Astro., 1986.

\bibitem{Massoumnia1986}
M.~A. Massoumnia, ``A geometric approach to the synthesis of failure detection
  filters,'' {\em IEEE Transactions on Automatic Control}, vol.~31, no.~9,
  pp.~839--846, 1986.

\bibitem{Massoumnia1989}
M.~A. Massoumnia, G.~C. Verghese, and A.~Willsky, ``Failure detection and
  identification,'' {\em IEEE Transactions on Automatic Control}, vol.~34,
  pp.~316--321, March 1989.

\bibitem{DrMeskin_TacFullPaper}
N.~Meskin and K.~Khorasani, ``A geometric approach to fault detection and
  isolation of continuous-time {M}arkovian jump linear systems,'' {\em IEEE
  Transactions on Automatic Control}, vol.~55, no.~6, pp.~1343--1357, 2010.

\bibitem{DrMeskin_Delay}
N.~Meskin and K.~Khorasani, ``Fault detection and isolation of distributed
  time-delay systems,'' {\em IEEE Trans. Automatic Control}, vol.~54,
  pp.~2680--2685, 2009.

\bibitem{DrMeskin_Impulsive}
N.~Meskin and K.~Khorasani, ``Fault detection and isolation of linear impulsive
  systems,'' {\em IEEE Transactions on Automatic Control}, vol.~56, no.~8,
  pp.~1905--1910, 2011.

\bibitem{IsidoriFDI}
C.~De~Persis and A.~Isidori, ``A geometric approach to nonlinear fault
  detection and isolation,'' {\em Transactions on Automatic Control}, vol.~46,
  no.~6, pp.~853--865, 2001.

\bibitem{MKR_SysTech2010}
N.~Meskin, K.~Khorasani, and C.~A. Rabbath, ``A hybrid fault detection and
  isolation strategy for a network of unmanned vehicles in presence of large
  environmental disturbances,'' {\em IEEE Transactions on Control Systems
  Technology}, vol.~18, no.~6, pp.~1422--1429, 2010.

\bibitem{MKR_Nonlin2010}
N.~Meskin, K.~Khorasani, and C.~A. Rabbath, ``Hybrid fault detection and
  isolation strategy for non-linear systems in the presence of large
  environmental disturbances,'' {\em IET control theory \& applications},
  vol.~4, no.~12, pp.~2879--2895, 2010.

\bibitem{ntogramatzidis2012GeometryArticle}
L.~Ntogramatzidis and M.~Cantoni, ``Detectability subspaces and observer
  synthesis for two-dimensional systems,'' {\em Multidimensional Systems and
  Signal Processing}, pp.~1--18, 2012.

\bibitem{ntogramatzidis2012Siam}
L.~Ntogramatzidis, ``Structural invariants of two-dimensional systems,'' {\em
  SIAM Journal on Control and Optimization}, vol.~50, no.~1, pp.~334--356,
  2012.

\bibitem{Valcher2013}
T.~Vasilache and V.~Prepelita, ``Observability and geometric approach of 2{D}
  hybrid systems,'' in {\em The International Conference on Systems, Control,
  Signal Processing and Informatics, Informatics}, (Greece), pp.~207--214,
  July, 2013.

\bibitem{2DLyapunov}
T.~Ooba, ``On stability analysis of 2-{D} systems based on 2-{D} lyapunov
  matrix inequalities,'' {\em IEEE Transactions on Circuits and Systems I:
  Fundamental Theory and Applications}, vol.~47, no.~8, pp.~1263--1265, 2000.

\bibitem{Controlability2D}
E.~Rogers, K.~Galkowski, A.~Gramacki, J.~Gramacki, and D.~Owens, ``Stability
  and controllability of a class of 2-{D} linear systems with dynamic boundary
  conditions,'' {\em IEEE Transactions on Circuits and Systems I: Fundamental
  Theory and Applications}, vol.~49, no.~2, pp.~181--195, 2002.

\bibitem{2DKalman}
Y.~Zou, M.~Sheng, N.~Zhong, and S.~Xu, ``A generalized {K}alman filter for 2{D}
  discrete systems,'' {\em Circuits, Systems and Signal Processing}, vol.~23,
  no.~5, pp.~351--364, 2004.

\bibitem{ACC2014}
A.~Baniamerian, N.~Meskin, and K.~Khorasani, ``Fault detection and isolation of
  fornasini-marchesini 2{D} systems: A geometric approach,'' in {\em 2014
  American Control Conference}, pp.~5527--5533, 2014.

\bibitem{Malek_3DFDI}
S.~Maleki, P.~Rapisarda, L.~Ntogramatzidis, and E.~Rogers, ``Failure
  identification for 3{D} linear systems,'' {\em Multidimensional Systems and
  Signal Processing}, vol.~26, no.~2, pp.~481--502, 2015.

\bibitem{Malek_3DFDIConf}
S.~Maleki, P.~Rapisarda, L.~Ntogramatzidis, and E.~Rogers, ``A geometric
  approach to 3{D} fault identification,'' in {\em Proceedings of the 8th
  International Workshop on Multidimensional Systems}, pp.~1--6, 2013.

\bibitem{Bisiacco_Obs}
M.~Bisiacco, ``On the state reconstruction of 2{D} systems,'' {\em Syst. Contr.
  Lett}, vol.~5, pp.~347--353, 1985.

\bibitem{Roess}
R.~P. Roesser, ``A discrete state-space model for linear image processing,''
  {\em IEEE Transaction on Automatic Control}, vol.~AC-20, pp.~1--10, 1975.

\bibitem{FM_SeparationSet}
E.~Fornasini and G.~Marchesini, ``Stability analysis of 2-{D} systems,'' {\em
  IEEE Transactions on Circuits and Systems}, vol.~27, no.~12, pp.~1210--1217,
  1980.

\bibitem{Zwart_Book}
{\em Geometric Theory for Infinite Dimensional Systems}, vol.~115.
\newblock Springer, 1989.

\bibitem{LMI1D}
P.~Gahinet and P.~Apkarian, ``A linear matrix inequality approach to
  ${H}_\infty$ control,'' {\em International Journal of Robust and Nonlinear
  Control}, vol.~4, no.~4, pp.~421--448, 1994.

\bibitem{Curtain_1986}
R.~F. Curtain, ``Invariance concepts in infinite dimensions,'' {\em SIAM
  Journal on Control and Optimization}, vol.~24, no.~5, pp.~1009--1030, 1986.

\bibitem{Wonham_Book}
W.~M. Wonham, {\em Linear multivariable control: a geometric approach}.
\newblock Springer-Verlag, second~ed., 1985.

\bibitem{conte1988GeometryConf}
G.~Conte and A.~Perdon, ``On the geometry of 2{D} systems,'' in {\em IEEE
  International Symposium on Circuits and Systems}, pp.~97--100, 1988.

\bibitem{conte1988GeometryArticle}
G.~Conte and A.~Perdon, ``A geometric approach to the theory of 2-{D}
  systems,'' {\em IEEE Transactions on Automatic Control}, vol.~33,
  pp.~946--950, 1988.

\bibitem{ECC2014}
A.~Baniamerian, N.~Meskin, and K.~Khorasani, ``Fault detection and isolation of
  riesz spectral systems: A geometric approach,'' in {\em 2014 European Control
  Conference}, pp.~2145--2152, June 2014.

\bibitem{FMMinimalRealization}
E.~Fornasini and G.~Marchesini, ``On the problems of constructing minimal
  realizations for two-dimensional filters,'' {\em IEEE Transactions on Pattern
  Analysis and Machine Intelligence}, no.~2, pp.~172--176, 1980.

\bibitem{HEleakageThesis}
G.~Szederkenyi, ``Simultaneous fault detection of heat exhangers,'' Master's
  thesis, University of Veszprem, 1998.

\bibitem{SiamPDEBook}
J.~C. Strikwerda, {\em Finite Difference Schemes and Partial Differential
  Equations, Second Edition}.
\newblock Society for Industrial and Applied Mathematics, 2th~ed., 2004.

\bibitem{BoundaryControl_1DApp}
A.~Maidi, M.~Diaf, and J.-P. Corriou, ``Boundary geometric control of a
  counter-current heat exchanger,'' {\em Journal of Process Control}, vol.~19,
  no.~2, pp.~297--313, 2009.

\bibitem{ThresholdCal}
S.~Ding, P.~Zhang, P.~Frank, and E.~L. Ding, ``Threshold calculation using
  {LMI}-technique and its integration in the design of fault detection
  systems,'' in {\em Proceedings of the IEEE Decision and Control Conference},
  pp.~469--474, 2003.

\bibitem{MonteCarloBook}
R.~Y. Rubinstein and D.~P. Kroese, {\em Simulation and the Monte Carlo method},
  vol.~707.
\newblock John Wiley \& Sons, 2011.

\end{thebibliography}

\end{document}